\documentclass[11pt,draftclsnofoot,onecolumn]{IEEEtran}
\usepackage{mathrsfs}
\usepackage{psfrag}
\usepackage{graphicx}
\usepackage{fancyhdr,epsfig}
\usepackage{mathtools}
\usepackage{research5IEEE}
\usepackage{amsmath}
\usepackage{amsfonts}
\usepackage{amssymb}
\usepackage{amsthm}
\newtheorem{Theorem}{Theorem}
\newtheorem{Lemma}{Lemma}
\newtheorem{Proposition}{Proposition}
\newtheorem{corollary}{Corollary}
\newtheorem{definition}{Definition}
\newtheorem{remark}{Remark}

\begin{document}
\title{Insufficiency of Linear-Feedback Schemes In Gaussian Broadcast Channels with Common Message}

\author{Youlong Wu, Paolo Minero, Mich\`ele Wigger%
\thanks{This paper was in part presented at the \emph{IEEE International Workshop on Signal Processing Advances for Wireless Communications}, in Darmstadt, Germany, June 2013. }
\thanks{Y. Wu and M. Wigger are with the Department of Communications and Electronics, Telecom Paristech, Paris, France. (e-mail: youlong.wu@telecom-paristech; michele.wigger@telecom-paristech.fr). Paolo Minero is with the Department of Electrical Engineering at University of Notre Dame.}%
\thanks{The work of Y. Wu and M. Wigger was partially supported by the city of Paris under program ``Emergences``.}
}

\date{\today}
\maketitle

\begin{abstract}
We consider the $K\geq 2$-user  memoryless Gaussian broadcast channel (BC) with feedback and common message only. We show that linear-feedback schemes with a message point, in the spirit of Schalkwijk\&Kailath's scheme for point-to-point channels or Ozarow\&Leung's scheme for BCs with private messages, are strictly suboptimal for this setup. Even with perfect feedback, the largest rate achieved by these schemes is strictly smaller than capacity $C$ (which is the same with and without feedback). In the extreme case where the number of receivers $K\to \infty$, the largest rate achieved by linear-feedback schemes with a message point tends to 0.

To contrast this negative result, we describe a  scheme for \emph{rate-limited} feedback that uses the feedback  in an intermittent way, i.e., the receivers send feedback signals only in few channel uses.  This scheme achieves all rates $R$ up to capacity $C$ with an $L$-th order exponential decay of the probability of error if the feedback rate $R_{\textnormal{fb}}$ is at least $(L-1)R$  for some positive integer $L$.

\end{abstract}


\section{Introduction}
We consider the $K\geq 2$-user Gaussian broadcast channel (BC) where the transmitter sends a single common message to all receivers. For this setup, even perfect feedback cannot increase capacity. Feedback can however  reduce the minimum probability of error for a given blocklength.

For Gaussian point-to-point channels \cite{scha,kim} or for memoryless Gaussian networks such as the multiple-access channel (MAC) \cite{ozarMac} and the BC with private messages \cite{ozar}, perfect feedback allows to have a double-exponential decay of the probability of error in the blocklength. These super-exponential decays of the probability of error are achieved by Schalkwijk\&Kailath type schemes that first map the  message(s) into real message point(s) and then send as their channel inputs linear combinations of the message point(s) and the past feedback signals. We call such schemes \emph{linear-feedback schemes with message points} or  \emph{linear-feedback schemes} for short.  Such schemes are known to achieve the capacity of Gaussian point-to-point channels (memoryless or with memory) \cite{scha,kim} and the sum-capacity of the two-user memoryless Gaussian MAC \cite{ozarMac}. For $K\geq 3$-user Gaussian MACs they are optimal among a large class of schemes \cite{Kramer,ardesta}, and for Gaussian BCs with private messages, they achieve the largest sum-rates known to date \cite{minero,wigger,wigger2011}.

In this paper we show that linear-feedback schemes with a message point are strictly suboptimal for the $K$-user memoryless Gaussian BC with common message and fail to achieve capacity. 
As a consequence, for this setup, linear-feedback schemes also fail to achieve double-exponential  decay of the probability of error for rates close to capacity. To our knowledge, this is the first example of a memoryless Gaussian network with perfect feedback, where linear-feedback schemes with message points are shown to be strictly suboptimal. In all previously studied networks with perfect feedback, they attained the optimal performance or the best so far performance. (In case of noisy feedback, they are known to perform badly even in the memoryless Gaussian point-to-point channel \cite{kimlap}.)

In the asymptotic scenario of infinitely many receivers $K\to \infty$, the performance of linear-feedback schemes with a message point even collapses completely: the largest rate that is achievable with these schemes tends to 0 as $K \to \infty$. This latter result holds under some mild assumptions regarding the  variances of the noises experienced at the receivers, which are for example met when all the noise variances are equal. Notice that, in contrast, the capacity of the $K$-user Gaussian BC with common message does not tend to 0 as $K\to \infty$ when e.g., all the noise variances are equal. In this case, the capacity does not depend on $K$, because it is  simply given by the point-to-point capacity to the receiver with the largest noise variance.

That the performance of linear-feedback schemes with a common message point degenerates with increasing number of users $K$  is intuitively explained as follows. At each time instant, the transmitter sends a linear combination of the message point and 
past noise symbols. Resending the noise symbols previously experienced at some Receiver~$k$ can be beneficial for this Receiver~$k$ because it  allows it to 
mitigate the noise corrupting previous outputs. However, resending these noise symbols is of no benefit for all other Receivers~$k'\neq k$ and only harms them. 
Therefore, the more receivers there are, the more noise symbols the transmitter sends in each channel use that are useless for a given Receiver~$k$.

For the memoryless Gaussian point-to-point channel \cite{scha} and MAC \cite{ozar}, the (sum-)capacity achieving linear-feedback schemes with message points transmit in each channel use a scaled version of the linear minimum mean square estimation (LMMSE) errors  of the message points given the previous channel outputs. The same strategy is however strictly suboptimal---even among the class of linear-feedback schemes with message points---when sending private messages over a Gaussian BC~\cite{minero}. It is unknown whether  LMMSE estimates are optimal among linear-feedback schemes when sending a common message over the  Gaussian BC. 


In our proof that any linear-feedback scheme with a message point cannot achieve the capacity of the Gaussian BC with common message, the following proposition is key: For any sequence of linear-feedback schemes with a \emph{common} message point that achieves rate $R>0$, one can construct a sequence of linear-feedback schemes that achieves the rate tuple $R_1= \ldots= R_K = R$ when sending  $K$ \emph{private} message points with a linear-feedback scheme.
This proposition shows that the class of linear-feedback schemes with message points cannot take advantage of the fact that all the $K\geq 2$ receivers are interested in the same message.

To contrast the bad performance of linear-feedback schemes, we present a coding scheme that uses the feedback in a intermittent way (that  only in few time slots the receivers send  feedback signals) \cite{reza} and that achieves double-exponential decay of the probability of error for all rates up to capacity. In our scheme it suffices to have rate-limited feedback with feedback rate $R_{\textnormal{fb}}$ no smaller than the forward rate $R$. If the feedback rate $R_{\textnormal{fb}}< R$ then, even for the setup with only one receiver, the probability of error can decay only exponentially in the blocklength \cite{reza}. This implies immediately that also for the $K\geq 2$ receivers BC with common message no double-exponential decay in the probability of error is achievable  when $R_{\textnormal{fb}}< R$. When the feedback rate $R_{\textnormal{fb}}>(L-1)R$, for some positive integer $L$, then our intermittent-feedback scheme can achieve an $L$-th order exponential decay in the probability of error. That means, it achieves a probability of error of the form $P_\textnormal{e}^{(n)} = \exp(-\exp(\exp(\ldots\exp(\Omega(n)))))$, where there are $L$ exponential terms and where $\Omega(n)$ denotes a function that satisfies $\varliminf_{n\to\infty}\frac{\Omega(n)}{n} >0$.

In our intermittent-feedback scheme communication takes place in $L$ phases. In the first phase, the transmitter uses a Gaussian code of power $P$ to send the common message to the $K$ Receivers. The transmission in phase $l\in\{2,\ldots, L\}$ depends on the feedback signals. After each phases $l\in\{1,\ldots, L-1\}$ each Receiver~$k$ feeds back a temporary guess of the message. Now, if one receiver's temporary guesses  after phase $(l-1)$ is wrong, then in phase $l$ the transmitter resends the common message using a new code. If all receivers' temporary guesses after phase $(l-1)$ were correct, in phase $l$ the transmitter sends the all-zero sequence.  In this latter case, no power is consumed in phase $l$. The receivers' final guess is their temporary guess after phase $L$.

That the described scheme can achieve an $L$-th order decay of the probability of error, roughly follows from the following inductive argument. Assume that the probability of the event
``one of the receivers' guesses is wrong after phase $l$", for $l\in\{1,\ldots, L-1\}$, has an $l$-th order exponential decay in the blocklength. Then, when sending the common message in phase $l+1$, the transmitter can use power that is $l$-th order exponentially large in the blocklength 
without violating the \emph{expected} average blockpower constraint. With such a code, in turn, the probability that after phase $l+1$ one of the receivers has a wrong guess can have an $(l+1)$-th order exponential decay in the blocklength.

The rest of the paper is organized as follows. This section is concluded with some remarks on notation.
Section~\ref{sec:system} describes the Gaussian BC with common message and defines the class of linear-feedback schemes with a message point. Section~\ref{sec: private} introduces the Gaussian BC with private messages and defines the class of linear-feedback schemes with private  message points. Section~\ref{sec:results} presents our main results. Finally, Sections~\ref{sec:proposition1} and~\ref{sec:theorem2} contain the proofs of our Theorems~\ref{the0} and \ref{thm:2}.

\emph{Notation:}
Let $\mathcal{R}$ denote the set of reals and $\mathbb{Z}^+$ the set of positive integers. Also, let $\set{K}$ denote the discrete set $\set{K}:=\{1,\ldots, K\}$, for some $K\in\mathbb{Z}^+$.  For a finite set $\set{A}$, we denote by $|\set{A}|$ its cardinality and by $\set{A}^j$, for $j\in\mathbb{Z}^+$, its $j$-fold Cartesian product,  $\set{A}^{j}:=\set{A}_1\times\ldots\times\set{A}_j$.

We use capital letters to denote  random variables and small letters for their realizations, e.g. $X$ and $x$. For  $j\in\mathbb{Z}^+$, we use the short hand notations $X^j$ and $x^j$ for the tuples $X^j:=(X_1,\ldots, X_j)$ and $x^j:=(x_1,\ldots, x_j)$. Vectors are displayed in boldface, e.g.,  $\vect{X}$ and $\vect{x}$ for a random and deterministic vector. Further, $|\cdot|$ denotes the modulus operation for scalars and $\|\cdot\|$ the norm operation for vectors. For matrices we use the font $\mat{A}$, and we use 
 $\|\mat{A}\|_F$ to denote its Frobenius norm.

 The abbreviation i.i.d. stands for \emph{independent and identically distributed}. All logarithms are taken with base $e$, i.e.,  $\log(\cdot)$ denotes the natural logarithm. We denote by
  $\set{Q}(\cdot)$  the tail probability of the standard normal distribution. The operator $\circ$ is used to denote function composition.

  We use the Landau symbols:  $\Omega(n)$ denotes any function that satisfies $\varliminf_{n\to\infty}\frac{\Omega(n)}{n} >0$ and $o(1)$ denotes any function that tends to 0 as $n\to \infty$.

\section{Setup}\label{sec:system}

\subsection{System Model and Capacity}\label{sec:modelBC}
\begin{figure}[!t]
\centering
\includegraphics[width=0.45\textwidth]{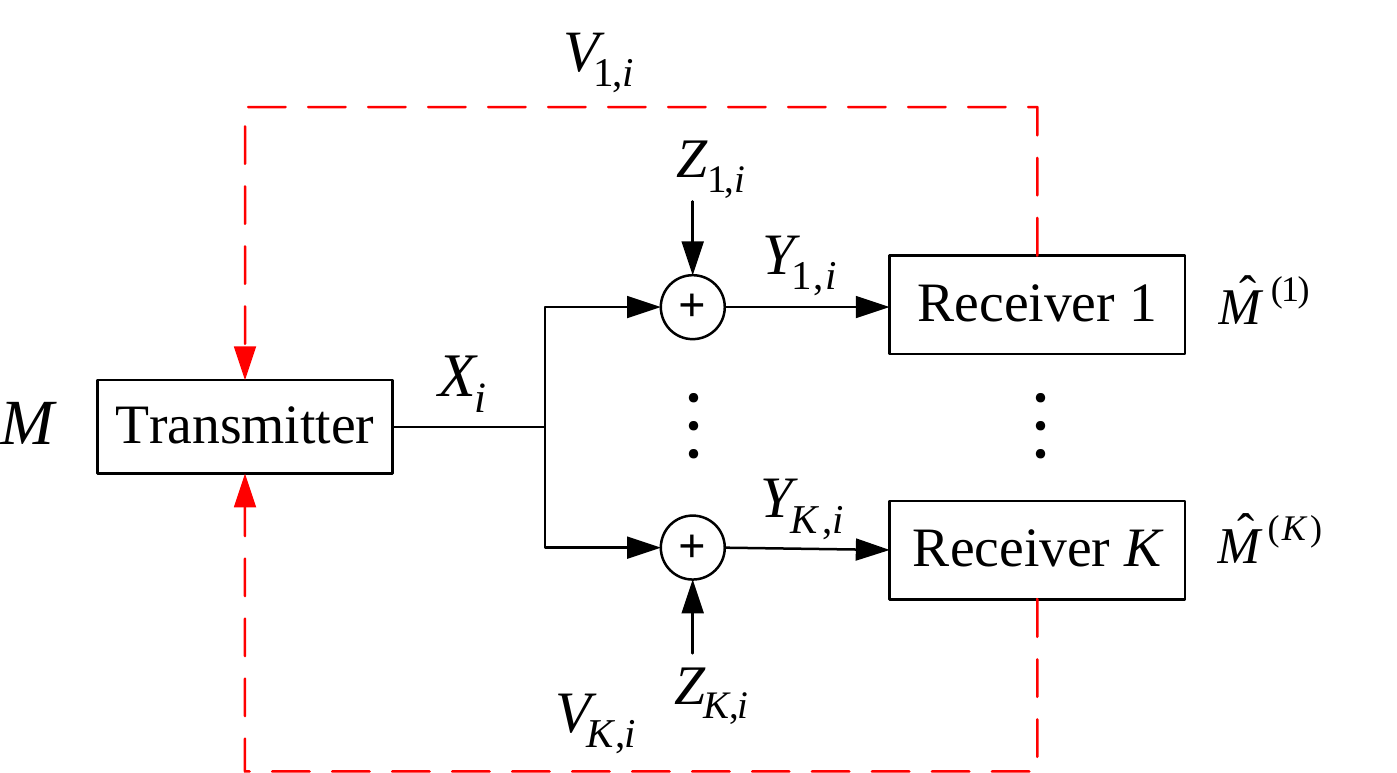}
\caption{$K$-receiver Gaussian Broadcast channel with feedback and common message only.} \label{fig:GausBC}
\vspace{-5mm}
\end{figure}

We consider the $K\geq 2$-receiver Gaussian BC with common message and feedback depicted in Figure~\ref{fig:GausBC}. Specifically, if  $X_i$ denotes the transmitter's channel input at time-$i$, the time-$i$ channel output at Receiver~$k\in\set{K}$ is
\begin{equation}\label{eq:BCk}
Y_{k,i}=X_{i}+Z_{k,i}
\end{equation}
where  $\{Z_{k,i}\}_{i=1}^n$ models the additive noise at Receiver~$k$. The sequence of noises $\{(Z_{1,i}, \ldots,Z_{K,i})\}_{i=1}^n$ is a sequence of i.i.d. centered Gaussian vectors, each of diagonal covariance matrix
\begin{equation}
\mat{K}_z=\begin{pmatrix}
\sigma_1^2  & \cdots &0\\ \vdots & \ddots& \vdots\\ 0 & \cdots &\sigma_K^2
\end{pmatrix}.
\end{equation}
Without loss of generality, we  assume that
\begin{equation}\label{eq:userorder}
\sigma_1^2 \geq \sigma_2^2 \geq \ldots \geq \sigma_K^2.
\end{equation}

The transmitter wishes to convey a common message  $M$ to all receivers, where $M$ is uniformly distributed over the message set $\mathcal{M}:= \{1,...,\lfloor e^{nR}\rfloor \}$ independent of the noise sequences $ \{ Z_{1,i}\}_{i=1}^n,\ldots, \{ Z_{K,i}\}_{i=1}^n$. Here, $n$ denotes the blocklength and $R> 0$ the rate of transmission.
It is assumed that the transmitter has either \emph{rate-limited} or \emph{perfect} feedback from all receivers.
That means, after each channel use $i\in\{1,\ldots, n\}$, each Receiver~$k\in\set{K}$  feeds back a signal $V_{k,i}\in\set{V}_{k,i}$ to the transmitter. The feedback alphabet $\set{V}_{k,i}$ is a design parameter of the scheme. In the case of rate-limited feedback, the signals from Receiver~$k$ have to satisfy:
\begin{equation}\label{eq:fbconstra}
\sum_{i=1}^n H(V_{k,i}) \leq n R_\textnormal{fb},~\quad k\in\set{K}
\end{equation}
where $R_{\textnormal{fb}}$ denotes the symmetric feedback rate. In the case of perfect feedback, we have no constraint on the feedback signals $\{V_{k,i}\}_{i=1}^n$, and it is thus optimal to choose $\set{V}_{k,i}=\Reals$ and
\begin{equation}
V_{k,i}=Y_{k,i},
\end{equation}
because in this way any processing that can be done at the receivers can also be done at the transmitter.

 An encoding strategy is comprised of a sequence of encoding functions $\{f^{(n)}_i\}^n_{i=1}$ of the form
\begin{IEEEeqnarray}{rCl}\label{eq:e}
f_i^{(n)} \colon   \set{M} && \times \set{V}_{1}^{i-1} \times\ldots \times \set{V}_{K}^{i-1} \to  \Reals\IEEEeqnarraynumspace
\end{IEEEeqnarray}
that is used to produce the channel inputs as
\begin{equation}\label{encf}
X_i=f^{(n)}_i(M,V^{i-1}_1,\ldots,V^{i-1}_K), \qquad i\in\{1,\ldots, n\}.
\end{equation}
We impose an expected average block-power constraint $P$ on the channel input sequence:
\begin{equation}\label{eq:power}
   \frac{1}{n}\mathbb{E}\left[\sum_{i=1}^{n}X^2_i\right] \leq P.
\end{equation}

Each Receiver~$k\in\set{K}$ decodes the message $M$ by means of a decoding function $g^{(n)}_k$
of the form
\begin{equation}\label{eq:dec}
g^{(n)}_k \colon \Reals^n \to \set{M}.
\end{equation}
That means, Receiver~$k$ produces as its guess
\begin{equation}
\hat{M}^{(k)}=g^{(n)}_k(Y^n_k).
\end{equation}

An error occurs in the communication if
\begin{equation}\label{eq:error}
(\hat{M}^{(k)}\neq M),
\end{equation}
for some $k\in\set{K}$.
Thus, the average probability of error is
\begin{equation}
 P_e^{(n)}:= \Prv{  \bigcup_{k\in\set{K}} \left( \hat{M}^{(k)}\neq M\right)}.
\end{equation}

We say that \emph{a rate $R>0$ is  achievable} for the described setup if for every $\epsilon>0$ there exists a sequence of encoding and decoding functions $\big\{ \{f_{i}^{(n)}\}_{i=1}^n, \{g_k^{(n)}\}_{k=1}^K\big\}_{n= 1}^{\infty}$ as in \eqref{eq:e} and \eqref{eq:dec} and satisfying the power constraint~\eqref{eq:power} such that for sufficiently large blocklengths $n$ the probability of error $P^{(n)}_{e}<\epsilon$. The supremum of all achievable rates is called the \emph{capacity}. The capacity is the same in the case of perfect feedback, of rate-limited feedback (irrespective of the feedback rate $R_\textnormal{fb}$), and without feedback.  We denote it by $C$ and by assumption~\eqref{eq:userorder} it is given by
\begin{IEEEeqnarray}{rCl}
C = \frac{1}{2} \log\left(1+\frac{P}{\sigma^2_1}\right) .\IEEEeqnarraynumspace
\end{IEEEeqnarray}

Our main interest in this paper is in the speed of decay of the probability of error at rates $R<C$.
\begin{definition} Given a positive integer $L$, we say that the $L$-th order exponential decay in the probability of error is achievable at a given rate $R<C$, if there exists a sequence of schemes of rate $R$ such that their probabilities of error $\{P_e^{(n)}\}_{n=1}^\infty$ satisfy
\begin{equation}\label{eq:exp}
\varliminf_{n\to \infty} \frac{1}{n} \log  \log \ldots \log (-\log  P_{e}^{(n)} )>0,
\end{equation}
where  the number of logarithms in \eqref{eq:exp} is $L$.
\end{definition}

\subsection{Linear-Feedback Schemes with a Message Point}
When considering perfect feedback, we will  be interested in the class of coding schemes where the feedback is only used in a linear fashion. Specifically, we say that a scheme is a linear-feedback scheme with a message point, if the sequence of encoding functions $\{f_i^{(n)}\}_{i=1}^n$ is of the form
\begin{equation}f_i^{(n)}= \Phi^{(n)} \circ L_i^{(n)}
\end{equation}with
\begin{subequations}\label{eq:linearfb0}
\begin{IEEEeqnarray}{rCl}
\Phi^{(n)}&\colon& M \mapsto \Theta^{(n)} \in \Reals\\
L_{i}^{(n)} &\colon& (\Theta^{(n)}, Y_1^{i-1}, \ldots,Y_K^{i-1}) \mapsto X_i
\end{IEEEeqnarray}
\end{subequations}
where $\Phi^{(n)}$ is an arbitrary function on the respective domains and $L_i^{(n)}$ is a \emph{linear mapping} on the respective domains. There is no constraint on the decoding functions $g_1^{(n)}, \ldots, g_K^{(n)}$.

By the definition of a linear-feedback coding scheme in~\eqref{eq:linearfb0}, for each blocklength $n$,  if we define $\vect{X}= \trans{(X_1, \ldots, X_n)}$, $\vect{Y}_k= \trans{(Y_{k,1},\ldots, Y_{k,n})}$, and $\vect{Z}_k= \trans{(Z_{k,1},\ldots, Z_{k,n})}$,  for $k\in\set{K}$,  the channel inputs can be written as:
\begin{equation}
\vect{X}= \Theta^{(n)} \cdot\vect{d}^{(n)}+ \sum_{k=1}^{K} \mat{A}_k^{(n)} \vect{Z}_k,
\end{equation}
for some $n$-dimensional vector $\vect{d}^{(n)}$ and  $n$-by-$n$ strictly lower-triangular matrices $\mat{A}_{1}^{(n)}, \ldots, \mat{A}_{K}^{(n)}$. (The lower-triangularity of $\mat{A}_1^{(n)}, \ldots, \mat{A}_K^{(n)}$ ensures that the feedback is used in a strictly causal fashion.) Thus, for a given blocklength $n$, a linear-feedback scheme is described by the tuple
\begin{equation}\label{eq:tuplelinfb}
\Phi^{(n)}, \vect{d}^{(n)}, \mat{A}_1^{(n)}, \ldots, \mat{A}_K^{(n)}, g_1^{(n)},\ldots,  g_K^{(n)}.
\end{equation}
It satisfies the average block-power constraint~\eqref{eq:power} whenever
\begin{equation}\label{eq:powerABd}
\sum_{k=1}^K\|\mat{A}_k^{(n)}\|_F^2 \sigma_k^2+\|\vect{d}^{(n)}\|^2\E{|\Theta^{(n)}|^2}\leq nP.
\end{equation}

The supremum of all rates that are achievable with a sequence of linear-feedback schemes with a message point is denoted by $C^{\textnormal{(Lin)}}$.

\section{For comparison: Setup with Private Messages and Perfect Feedback}\label{sec: private}
\subsection{System Model and Capacity Region}
For comparison, we also discuss the scenario where the  transmitter wishes to communicate a  private message~$M_{k}$ to each Receiver~$k\in\set{K}$ over the Gaussian BC in Figure~\ref{fig:GausBC}. The messages $M_1, \ldots, M_K$ are assumed independent of each other and of the noise sequences $\{Z_{1,i}\}_{i=1}^n,\ldots, \{Z_{K,i}\}_{i=1}^n$ and each $M_k$ is uniformly distributed over the set $\set{M}_k:=  \{1,\ldots, \lfloor e^{n R_k} \rfloor\}$. For this setup we restrict attention to perfect feedback. Thus, here the channel inputs are produced as
\begin{equation}\label{eq:encdef2}
X_i=f^{(n)}_{\textnormal{priv},i}(M_1,\ldots, M_K,Y_1^{i-1},\ldots,Y_K^{i-1}), ~ i\in\{1,\ldots, n\}.
\end{equation}
Receiver~$k$ produces the guess
\begin{eqnarray}
\hat{M}_{k}&= &g_{\textnormal{priv},k}^{(n)}({Y}_{k}^n)
\end{eqnarray}
where the sequence of decoding function $\{g_{\textnormal{priv},k}^{(n)}\}^K_{k=1}$ is of the form
\begin{IEEEeqnarray}{rCl}
g_{\textnormal{priv},k}^{(n)} &\colon& \mathbb{R}^{n} \rightarrow \{1,\ldots, \lfloor e^{n R_k}
\rfloor\} ,\label{eq:decdef12}
\end{IEEEeqnarray}

A rate tuple $(R_1,\ldots,R_K)$ is said to be achievable if for every
blocklength $n$ there exists a set of $n$ encoding functions
 as in
\eqref{eq:encdef2} satisfying the power constraint \eqref{eq:power}
and a set of $K$ decoding functions as in
\eqref{eq:decdef12} such that the probability of
decoding error  tends to 0 as the blocklength $n$
tends to infinity, i.e.,
\begin{equation*}
\lim_{n\rightarrow \infty}\text{Pr}\left[(M_1,\ldots,M_K)\neq
(\hat{M}_{1},\ldots,\hat{M}_{K})\right] =0.
\end{equation*}
The closure of the set of all achievable rate tuples $(R_1,\ldots,R_K)$ is
called the \emph{capacity region}. We denote it $\set{C}_{\textnormal{private}}$. This capacity region is unknown to date. (The sum-capacity in the high-SNR asymptotic regime is derived in~\cite{wigger}.) Achievable regions were presented in~\cite{minero,wigger,wigger2011}; the tighest known outer bound on capacity for $K=2$ users was presented in~\cite{ozar} based on the idea of revealing one of the output sequences to the other receiver. This idea generalizes to $K\geq 2$ users, and leads to the following outer bound \cite{Kramer,Bergmans}:
\begin{Lemma}\label{lem:privateouterbound}
If the rate tuple $(R_1, \ldots, R_K)$ lies in $\set{C}_{\textnormal{private}}$, then there exist coefficients $\alpha_1,\ldots, \alpha_K$ in the closed interval $[0,1]$ such that for each $k\in\set{K}$,
\begin{IEEEeqnarray}{rCl}
R_k\leq \frac{1}{2} \log \left(1+ \frac{\alpha_k P}{(1- \alpha_1-\ldots- \alpha_k)P+N_k } \right)
\end{IEEEeqnarray}
where
\begin{equation}\label{eq:newnoise}
N_k = \left(\sum_{k'=1}^{k} \frac{1 }{\sigma_{k'}^2 }\right)^{-1}, \qquad k\in\set{K}.
\end{equation}
\end{Lemma}
\begin{proof} Let a genie reveal each output sequence $Y_k^n$ to Receivers~$k+1, \ldots, K$. The resulting BC is physically degraded, and thus its capacity is the same as without feedback~\cite{elgamal} and known. Evaluating this capacity region  readily gives the outer bound in the lemma.
\end{proof}

\subsection{Linear-Feedback Schemes with Message Points}

A \emph{linear-feedback scheme with message points} for this setup with independent messages consists of a sequence of $K$ decoding functions as in~\eqref{eq:decdef12} and of a sequence of encoding functions $\{f^{(n)}_{\textnormal{priv},i}\}_{i=1}^n$ of the form
\begin{equation}f_{\textnormal{priv},i}^{(n)}= \Phi^{(n)}_{\textnormal{priv}} \circ L_{\textnormal{priv},i}^{(n)}
\end{equation}with
\begin{subequations}\label{eq:linearfb}
\begin{IEEEeqnarray}{rCl}
\Phi^{(n)}_{\textnormal{priv}}&\colon& \begin{pmatrix} M_1\\ \vdots\\M_K\end{pmatrix} \mapsto \boldsymbol \Theta:= \begin{pmatrix}  \Theta_1 \\ \vdots\\ \Theta_K \end{pmatrix} \in \Reals^K\\
L_{\textnormal{priv},i}^{(n)} &\colon& (\boldsymbol \Theta, Y_1^{i-1},\ldots, Y_K^{i-1}) \mapsto X_i
\end{IEEEeqnarray}
\end{subequations}
where $\Phi^{(n)}_{\textnormal{priv}}$ is an arbitrary function on the respective domains and $L_{\textnormal{priv},i}^{(n)}$ is a \emph{linear mapping} on the respective domains.

We denote the closure of the set of rate tuples $(R_1, \ldots, R_K)$ that are achievable with a linear-feedback scheme with message points by $\set{C}_{\textnormal{private}}^{(\textnormal{Lin})}$. This region  is unknown to date.

\section{Main Results} \label{sec:results}

The main question we wish to answer is whether for the Gaussian BC with common message a super-exponential decay in the probability of error is achievable for all rates $R<C$.
We first show that the class of linear-feedback schemes with message point fails in achieving this goal even with perfect feedback, because  it does not achieve capacity (Theorem~\ref{the0} and Corollary~\ref{cor:1}). As the number of receivers $K$ increases, 
the largest rate that is achievable with linear-feedback schemes with a message point even vanishes (Proposition~\ref{prop:condnoi}). However, as we show then, a super-exponential decay in the probability of error is still possible by means of an intermittent feedback scheme similar to \cite{reza}  (Theorem~\ref{thm:2}).


\begin{Proposition}\label{the1}
If a sequence of linear-feedback schemes with a message point achieves a  \emph{common rate} $R>0$, then there exists a sequence of linear-feedback schemes with message points that achieves the \emph{private rates} $(R,\ldots,R)\in\Reals^{K}$:
\begin{equation}
0 < R \leq C^{(\textnormal{Lin})} \quad \Longrightarrow \quad {(R,\ldots,R)} \in \set{C}_{\textnormal{private}}^{(\textnormal{Lin})}.
\end{equation}
\end{Proposition}
\begin{proof}
See Section~\ref{sec:proposition1}.
\end{proof}
Proposition~\ref{the1} and the upper bound in Lemma~\ref{lem:privateouterbound} yield the following result:

\begin{Theorem} \label{the0}
We have:
\begin{IEEEeqnarray}{rCl}\label{eq:upperbound}
C^{\textnormal{(Lin)}} \leq \frac{1}{2}\log\left(1+ \frac{\alpha_1^\star P}{(1- \alpha_1^\star)P+\sigma_1^2}\right)
\end{IEEEeqnarray}
where $\alpha_1^\star$ lies in the open interval $(0,1)$ and is such that there exist $\alpha_2^\star,\ldots, \alpha_K^\star\in(0,1)$ that satisfy
\begin{subequations}\label{eq:alphastar}
\begin{IEEEeqnarray}{rCl}\label{eq:alphasum}
\alpha_1^\star+ \alpha_2^\star+ \ldots + \alpha_K^\star=1
\end{IEEEeqnarray}
and for $k\in\{2,\ldots, K\}$:
\begin{IEEEeqnarray}{rCl}\label{eq:Requal}
\lefteqn{
\frac{1}{2} \log\left( 1+ \frac{\alpha_k^\star P}{(1-\alpha_1^\star- \alpha_2^\star-\ldots- \alpha_{k}^\star)P+ N_k}\right)}\qquad \nonumber\\& = & \frac{1}{2} \log\left( 1+ \frac{\alpha_1^\star P}{(1-\alpha_1^\star) P+\sigma_1^2}\right) \hspace{2.4cm}
\end{IEEEeqnarray}
\end{subequations}
where the noise variances $\{N_k\}^K_{k=1}$ are defined in~\eqref{eq:newnoise}.
\end{Theorem}
Since $\alpha_1^\star$ is strictly smaller than $1$, irrespective of $K$ and the noise variances $\sigma_1^2, \ldots, \sigma_K^2$, we obtain the following corollary.
\begin{corollary} \label{cor:1}
Linear-feedback schemes with a message point cannot achieve the capacity of the Gaussian BC with common message:
\begin{equation}
C^{\textnormal{(Lin)}} < C
\end{equation}
where the inequality is strict.
\end{corollary}

\begin{Proposition}\label{prop:condnoi}
If the noise variances $\{\sigma_k^2\}_{k=1}^K$ are such that
\begin{equation}\label{eq:condnoisevariance}
\sum_{k=1}^\infty N_k =\infty,
\end{equation}
then 
\begin{equation}\label{eq:czero}
\lim_{K\to \infty} C^{\textnormal{(Lin)}}=0.
\end{equation}
\end{Proposition}
\begin{proof}
See Appendix~\ref{app1}.
\end{proof}
In Figure \ref{fig:CapVSupper} we plot the upper bond on $C^{(\textnormal{Lin})}$ shown in (\ref{eq:upperbound}), Theorem \ref{the0}, as a function of the number of receivers $K$, which have all the same noise variance $\sigma_1^2=\ldots= \sigma_K^2=1$. As we observe, this upper bound, and thus also $C^{(\textnormal{Lin})}$, tends to 0 as $K$ tends to infinity

\begin{figure}[!t]
\centering
\includegraphics[width=0.45\textwidth]{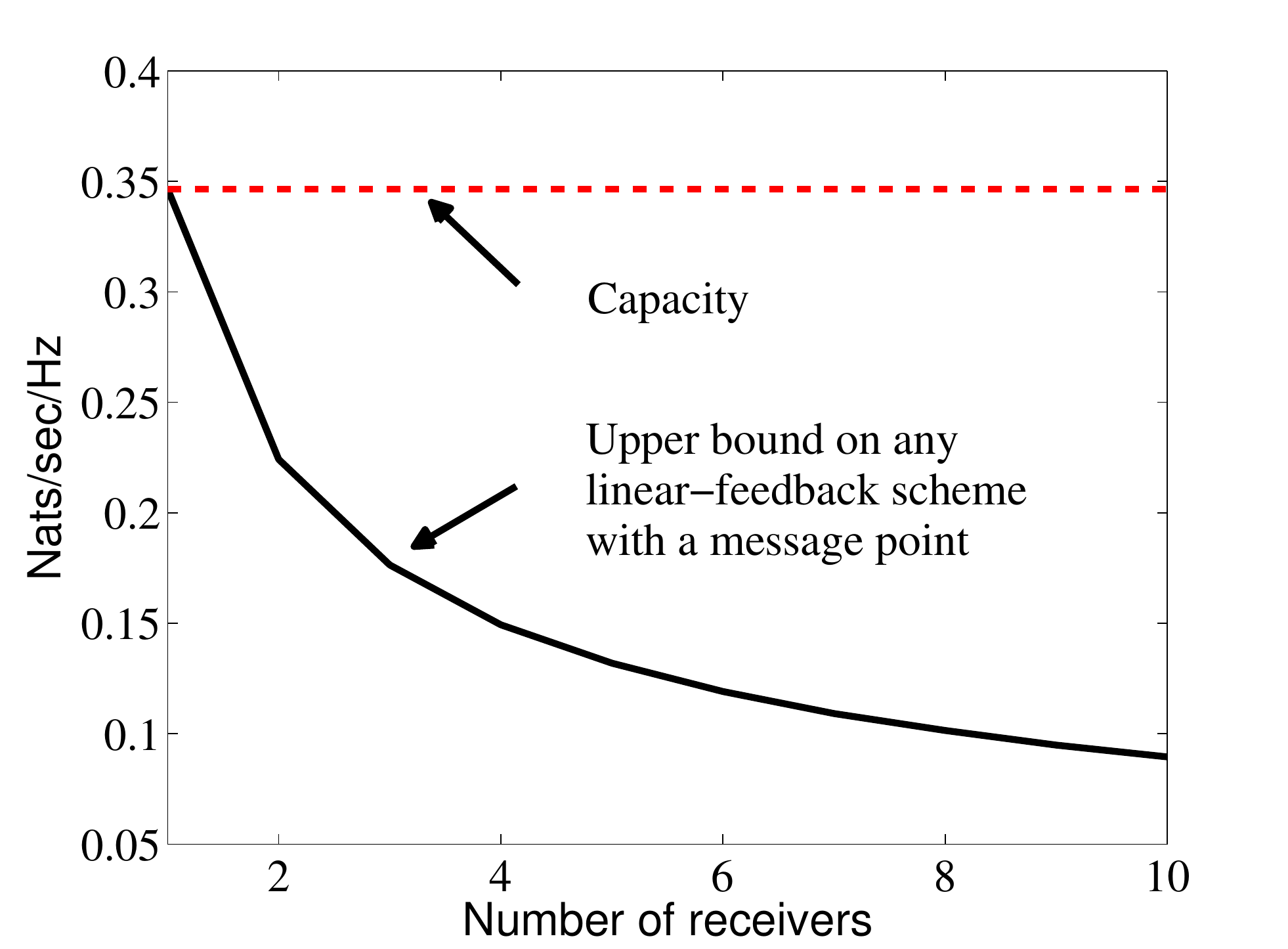}
\caption{Upper bound (\ref{eq:upperbound}) on the rates achievable with linear-feedback schemes with a message point in function of the number of receivers $K$.} \label{fig:CapVSupper}
\vspace{-5mm}
\end{figure}

%
\begin{Theorem}\label{thm:2}
For any positive rate  $R< C$, if the feedback rate
 \begin{equation}
 R_{\textnormal{fb}} \geq (L-1)R,
 \end{equation}
for some positive integer $L$, then  it is possible to achieve an $L$-th order exponential decay of the probability of error in the blocklength.
\end{Theorem}
\begin{proof}
See Section~\ref{sec:theorem2}.
\end{proof}

\section{Proof of Proposition~\ref{the1}}\label{sec:proposition1}

Let $\delta>0$ be a small real number.
Fix a sequence of rate-$R>0$, power-$(P-\delta)$ linear-feedback schemes that sends a common message point  over the Gaussian BC with probability of error $P_\textnormal{e}^{(n)}$ tending to 0 as $n\to\infty$. For each  $n\in\Integers^+$, let
\begin{equation}\label{eq:tuplesseq}
\Phi^{(n)}, \vect{d}^{(n)}, \mat{A}_1^{(n)}, \ldots, \mat{A}_K^{(n)},  g_1^{(n)}, \ldots, g_K^{(n)}
\end{equation}
denote the parameters of the blocklength-$n$ scheme, which satisfy the power constraint
 \begin{equation}\label{eq:dd}
\E{|\Theta^{(n)}|^2}\cdot \|\vect{d}^{(n)}\|^2+  \sum_{k=1}^K \big\|\mat{A}_k^{(n)}\big\|^2_F \leq  n(P-\delta)
 \end{equation}
 where $\Theta^{(n)}= \Phi^{(n)}(M)$.

We have the following lemma.
\begin{Lemma}\label{lem:vlog}
For each blocklength $n$, there exist $n$-dimensional row-vectors $\vect{v}_1^{(n)}, \ldots, \vect{v}_K^{(n)}$ of unit norms,
\begin{equation}\label{eq:vnorm}
\|\vect{v}_1^{(n)}\|^2=\cdots=\|\vect{v}_K^{(n)}\|^2 =1,
\end{equation}
and $K$ indices $j_1^{(n)}, \ldots, j_K^{(n)}\in\{1,\ldots, n\}$
such that for each $k\in\set{K}$ the following three limits holds:
\begin{enumerate}
\item
\begin{IEEEeqnarray}{rCl}\label{gamma10}
R&\leq & \varliminf\limits_{n \to \infty} - \frac{1}{2n} \log c_k^{(n)}
\end{IEEEeqnarray}
where
\begin{IEEEeqnarray}{rCl}\label{eq:cu}
c_k^{(n)} := \sigma_k^2 \big\|\vect{v}_k^{(n)}\big(\mat{I}+\mat{A}_k^{(n)}\big)\big\|^2+\!\!\!\! \sum_{ k'\in\set{K}\backslash \{k\}}\!\!\! \sigma_{k'}^2\big\|\vect{v}_{k'}^{(n)}\mat{A}_{k'}^{(n)}\big\|^2;\nonumber\\ 
\end{IEEEeqnarray}
\item
\begin{IEEEeqnarray}{rCl}
\lim_{n \to \infty} \frac{1}{n}\E{\left(X^{(n)}_{j_k^{(n)}}\right)^2} = 0 \label{Lemma1A}
\end{IEEEeqnarray}
where for $i\in\{1,\ldots, n\}$, $X^{(n)}_i$ denotes the $i$-th channel input of the blocklength-$n$ scheme; and
\item
\begin{IEEEeqnarray}{rCl}
\lim_{n \to \infty} \frac{1}{2n}\log\big(|v^{(n)}_{k,j_k^{(n)}}|\big)=0 \label{Lemma1C}
\end{IEEEeqnarray}
 where for $i\in\{1,\ldots, n\}$, $v_{k,i}^{(n)}$ denotes the $i$-th component of the vector $\vect{v}_k^{(n)}$.
 \end{enumerate}
\end{Lemma}
\begin{proof}
See Appendix~\ref{app:fano}.
\end{proof}

In the following, let for each $n\in\Integers^+$, $\vect{v}_1^{(n)}, \ldots, \vect{v}_K^{(n)}$ be $n$-dimensional unit-norm row-vectors  and $j_1^{(n)}, \ldots, j_K^{(n)}$ be positive integers satisfying the limits~\eqref{gamma10}, \eqref{Lemma1A}, and \eqref{Lemma1C}.

We now construct a sequence of linear-feedback schemes with message points that can send  $K$ independent  messages $M_1,\ldots,M_K$ to Receivers~$1,\ldots,K$ at rates
\begin{IEEEeqnarray}{rCl}\label{eq:rate}
R_k &\geq& \Big(\varliminf\limits_{n \to \infty} - \frac{1}{2n} \log c_k^{(n)}\Big)-\epsilon, \qquad k\in\set{K},
\end{IEEEeqnarray}
for an arbitrary small $\epsilon>0$;
with a probability of error that tends to 0 as the blocklength tends to infinity; and with an average blockpower that is no larger than $P$ when the blocklength is sufficiently large. By~\eqref{gamma10}, since $\delta, \epsilon>0$ can be chosen arbitrary small, and since $C^{\textnormal{(Lin)}}$ is continuous in the power $P$ (Remark~\ref{rem:region} ahead) and  is defined as a {supremum}, the result in Proposition~\ref{the1} will follow.

We describe our scheme for  blocklength-$(n+2K)$, for some fixed  $n\in\mathbb{Z}^+$. Our scheme is based on the  parameters $ \mat{A}_1^{(n)}, \ldots, \mat{A}_K^{(n)}$ in~\eqref{eq:tuplesseq}, on the vectors $\vect{v}_1^{(n)}, \ldots, \vect{v}_K^{(n)}$, and on the indices $j_1^{(n)}, \ldots, j_K^{(n)}$. For ease of notation, when describing our scheme in the following, we drop the superscript $(n)$, i.e., we write
\begin{equation*}
 \mat{A}_1, \ldots, \mat{A}_K, \; \vect{v}_1, \ldots, \vect{v}_K, \; \textnormal{ and } \; j_1, \ldots, j_K.
\end{equation*}
We also assume that
\begin{equation}
j_1\leq j_2\leq  \ldots\leq  j_K.
\end{equation}
(If this is not the case, we simply relabel the receivers.)
Also, to further simplify the description of the linear-feedback coding and the decoding, we rename the $n+2K$ transmission slots as depicted in Figure~\ref{fig:slots}.
Transmission starts at slot $1-K$ and ends at slot $n$; also, after each slot $j_k$, for $k\in\set{K}$, we introduce an additional slot $\tilde{j}_k$. We call the slots $1-K, \ldots, 0$ the \emph{initialization slots}, the slots $\tilde{j}_1,\ldots,\tilde{j}_K$ the \emph{extra slots}, and the remaining slots $1,2,3,\ldots, n$  the \emph{regular slots}.
\begin{figure*}[!t]
\centering
\includegraphics[width=0.95\textwidth]{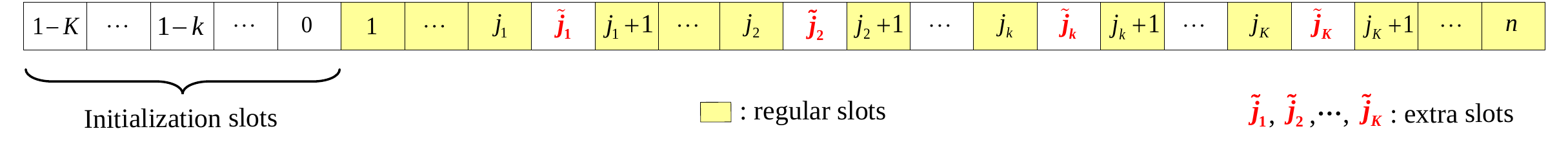}
\caption{Labeling of the transmission slots for our blocklength-$(n+2K)$ scheme.} \label{fig:slots}
\end{figure*}


In our scheme, the message points $\{ \Theta_k\}_{k=1}^K$ are constructed as in the Ozarow-Leung scheme \cite{ozar}:
\begin{equation}\label{eq:Theta1}
 \Theta_k := 1/2 - \frac{M_k-1}{\lfloor e^{(n+2K) R_k} \rfloor}, \qquad k\in\set{K}.
\end{equation}
These messages are sent during the initialization phase. Specifically, in the initialization slots $i=1-K,\ldots, 0$, the transmitter sends the $K$ message points $\Theta_1,\ldots, \Theta_K$:
\begin{equation}\label{eq:Xinit}
X_{1-k}= \sqrt{\frac{P}{\Var{\Theta_k}}} \Theta_{k}, \qquad k\in\set{K}.
\end{equation}
In the regular slots $i=1,\ldots,n$, the transmitter sends the same inputs as in the scheme with common message described by the parameters in~\eqref{eq:tuplesseq}, but without the component from the message point and  where for each $k\in\set{K}$ the noise sample $Z_{k,j_k}$ is replaced by $Z_{k,\tilde{j}_k}$. Thus, defining the $n$-length vector of regular inputs $\vect{X}\triangleq \trans{(X_1, X_2, X_3, \ldots, X_n)}$, we have
\begin{equation}\label{eq:inputsn}
\vect{X}= \sum_{k=1}^{K} \mat{A}_k \vect{\tilde{Z}}_k
\end{equation}
where for $k\in\set{K}$,
\begin{equation}
\vect{\tilde{Z}}_k := \trans{(Z_{k,1}, Z_{k,2}, \ldots, Z_{k,j_k-1}, Z_{k, \tilde{j}_k},  Z_{k,j_k+1}, \ldots, Z_{k,n})}
\end{equation}
denotes the $n$-length noise vector experienced at Receiver~$k$ during the regular slots $1,\ldots, j_k-1$,  the extra slot $\tilde{j}_k$, and  the regular slots $j_k+1,\ldots, n$.

Since for each $k\in\set{K}$, the extra slot $\tilde{j}_k$ preceds all regular slots $j_k+1, \ldots, n$, the strict lower-triangularity of the matrices $\mat{A}_1, \ldots, \mat{A}_K$ ensures that in~\eqref{eq:inputsn} the feedback is used in a strictly causal way.

In each extra slot $\tilde{j}_k$, for $k\in\set{K}$, the transmitter sends the regular input $X_{j_k}$, but now with the  noise sample $Z_{k, 1-k}$,
\begin{equation}\label{eq:Xbis}
X_{\tilde{j}_k}= X_{j_k} + Z_{k,1-k}.
\end{equation}
 The noise sample $Z_{k,1-k}$ is of interest to Receiver~$k$ (and only to Receiver~$k$) because from this noise sample and $Y_{k,1-k}$ one can recover $\Theta_k$, see~\eqref{eq:Xinit}. Therefore---as described shortly---in the decoding, Receiver~$k$ considers the extra output $Y_{k,\tilde{j}_{k}}$ which contains $Z_{k, 1-k}$  whereas all other receivers $k'\neq k$ instead consider  the regular outputs $Y_{k', j_k}$ which do not have the $Z_{k,1-k}$-component.

\begin{figure*}[!t]
\centering
\includegraphics[width=0.95\textwidth]{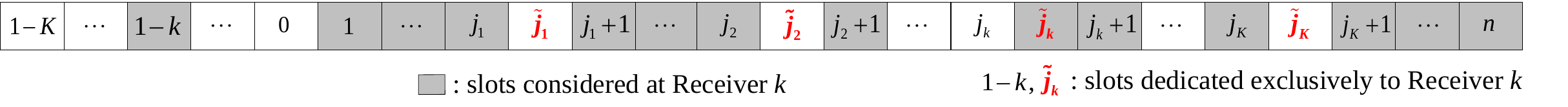}
\caption{Transmissions considered at Receiver~$k$ and transmissions dedicated exclusively to Receiver~$k$.} \label{fig:slots_Receiverk}
\end{figure*}
The decoding is similar as in the Ozarow-Leung scheme. However, here, each Receiver~$k\in\set{K}$ only considers the initialization output $Y_{k,1-k}$, the regular outputs $Y_{k,1},\ldots, Y_{k,j_k-1}, Y_{k,j_k+1},\ldots, Y_{k,K}$ and the extra output $Y_{k,\tilde{j_k}}$, see also Figure~\ref{fig:slots_Receiverk}. Specifically, Receiver~$k$ forms the $n$-length vector
\begin{IEEEeqnarray}{rCl}\label{eq:vecY}
\vect{\tilde{Y}}_k &:=& \trans{\big(Y_{k,1}, \ldots, Y_{k,j_k-1}, Y_{k,\tilde{j}_k}, Y_{k, j_k+1}, \ldots, Y_{k,n} \big) },\IEEEeqnarraynumspace
\end{IEEEeqnarray}
and produces the LMMSE estimate $\hat{Z}_{k,1-k}$ of the noise $Z_{k,1-k}$ based on the vector $\vect{\tilde{Y}}_k$. It then forms
\begin{IEEEeqnarray}{rCl}
\hat{\Theta}_k &=&  \sqrt{\frac{\mathrm{Var}(\Theta_k)}{P}}\left( Y_{k, 1-k} - \hat{Z}_{k,1-k}\right).
\end{IEEEeqnarray}
and performs nearest neighbor decoding to decode its desired Message $M_k$ based on $\hat{\Theta}_k$.

We now analyze the described scheme.
The expected blockpower of our scheme is:
\begin{IEEEeqnarray}{rCl}
\lefteqn{\sum_{i=1-K}^{0}\E{| X_i|^2}  + \sum_{i=1}^{n}\E{| X_i|^2}  + \sum_{k=1}^K \E{\big|X_{\tilde{j}_k}\big|} }\quad \nonumber \\
  & \leq & K P + n(P-\delta ) +\sum_{k=1}^K \E{|X_{j_k}|^2} +\sum_{k=1}^{K} \sigma_k^2\IEEEeqnarraynumspace
\end{IEEEeqnarray}
where the  inequality follows from~\eqref{eq:Xinit},  \eqref{eq:inputsn}, and \eqref{eq:Xbis}, and from~\eqref{eq:dd}, which assures that the regular inputs $X_{1}, \ldots, X_n$ are block-power constrained to $n(P-\delta)$.  Further, since  the  indices $j_1, \ldots, j_K$ satisfy Assumption~\eqref{Lemma1A},
\begin{equation}
\lim_{n\to \infty} \frac{1}{n} \sum_{k=1}^K \E{|X_{j_k}|^2}=0,
\end{equation}
and thus for sufficiently large $n$ the proposed scheme for independent messages is average blockpower constrained to $P$.

We analyze the probability of error. Notice that
\begin{IEEEeqnarray}{rCl}
\hat{\Theta}_k & = & \Theta_k+ E_{k}
\end{IEEEeqnarray}
where
\begin{equation}
E_k:= \sqrt{\frac{\mathrm{Var}({\Theta}_k)}{P}}\left( {Z}_{k,1-k} - \hat{Z}_{k,1-k}\right)
\end{equation}
is zero-mean Gaussian of variance
\begin{equation}\label{eq:Var1}
\Var{E_k}= \frac{\Var{\Theta_k}}{P} \sigma_k^2 e^{-2I\left(Z_{k,1-k}; \vect{\tilde Y}_k\right)}.
\end{equation}
 Equation~\eqref{eq:Var1} is justified by
 \begin{IEEEeqnarray}{rCl}
I\big(Z_{k,1-k}&; &\vect{\tilde Y}_k\big) \nonumber \\
& = & h( {Z}_{k,1-k} ) - h\left(Z_{k,1-k}\big| \vect{\tilde Y}_{k}\right) \nonumber\\
& = & \frac{1}{2} \log \left( \frac{ \sigma_k^2 }{ \Var{ {Z}_{k,1-k} - \hat{Z}_{k,1-k}}} \right) \label{eq:I1}
\end{IEEEeqnarray}  where the last equality follows because $Z_{k,1-k}$ and $\vect{\tilde{Y}}_k$ are jointly Gaussian, and thus the LMMSE estimation error $Z_{k,1-k}-\hat{Z}_{k,1-k}$ is independent of the observations $\vect{\tilde Y}_k$.

The nearest neighbor decoding rule is successful if $|E_{k}|$ is smaller than half the distance between any two message points. Since $E_{k}$ is Gaussian and independent of the message point, the probability of this happening is
\begin{IEEEeqnarray*}{rCl}
\Prv{\hat{M}_k\neq M_k}&\leq& \Prv{|E_k |\geq \frac{1}{2\cdot \lfloor e^{(n+2K)R_k}\rfloor }}\\
& =&2\set{Q} \left(\frac{e^{I(Z_{k, 1-k}; \vect{\tilde{Y}}_k)}}{2\cdot \lfloor e^{(n+2K)R_k}\rfloor}\cdot\frac{P}{\Var{{\Theta}_k}\sigma_k^2}\right).\IEEEeqnarraynumspace
\end{IEEEeqnarray*}
We conclude that the probability of error tends to 0, double-exponentially, whenever
\begin{IEEEeqnarray}{rCl}\label{rate1}
  R_k&<& \varliminf_{n \to \infty} \frac{1}{n}I(Z_{k,1-k}; \vect{\tilde{Y}}_k).
  \end{IEEEeqnarray}
  Notice that the vector $\vect{\tilde{Y}}_k$ as defined in~\eqref{eq:vecY}, satisfies
  \begin{IEEEeqnarray}{rCl}
\vect{\tilde{Y}}_k &=&
 \sum_{k'\in\set{K}\backslash\{k\}} \mat{A}_{k'} \vect{\tilde Z}_{k'} + (\mat{I} + \mat{A}_k) \vect{\tilde{Z}}_k + \vect{e}_{j_k} Z_{k, 1-k}\label{eq:Yuseful1} \IEEEeqnarraynumspace
\end{IEEEeqnarray}
where for each $i\in\{1,\ldots, n\}$ the vector $\vect{e}_i$ is the $n$-length unit-norm vector with all zero entries except at position $i$ where the entry is $1$. Thus, by the data processing inequality,
 \begin{IEEEeqnarray}{rCl}
\lefteqn{I(Z_{k,1-k}; \vect{\tilde{Y}}_k)}\nonumber\\ & \geq &I(Z_{k,1-k}; \trans{\vect{v}}_k \vect{\tilde{Y}}_k)\nonumber\\
 &=& \frac{1}{2}\log \left(1\!+\!\frac{|{v}_{k,j_k}|^2}{\sigma_k^2\|\vect{v}_k(\mat{I}+\mat{A}_k)\|^2+\!\!\!\!\sum\limits_{k'\in\set{K}\backslash \{k\}} \!\!\! \!\!\sigma_{k'}^2\|\vect{v}_{k'}\mat{A}_{k'}\|^2}\right) \nonumber\\
 &=& \frac{1}{2}\log \left(1+\frac{|{v}_{k,j_k}|^2}{c_k}\right)\label{eq:noli}
  \end{IEEEeqnarray}
where the first equality follows by~\eqref{eq:Yuseful1} and the joint Gaussianity of all involved random variables and the second equality follows by the definition of $c_k$ in~\eqref{eq:cu}.

Combining~\eqref{rate1} and \eqref{eq:noli}, we obtain that the probability $\Prv{\hat{M}_k\neq M_k}$  tends to 0 as $n\to \infty$ whenever
  \begin{IEEEeqnarray}{rCl}\label{rate11}
  R_k
  &<& \varliminf_{n \to \infty} \frac{1}{2n} \log\bigg(1+\frac{|{v}_{k,j_k}|^2}{c_k}
\bigg). \end{IEEEeqnarray}
(Recall that the quantities $j_k,$ $c_k$,  and $v_{k,j_k}$ depend on $n$, but here we do not show this dependence for readability.)

Further, by the converse in~\eqref{gamma10},
  \begin{IEEEeqnarray}{rCl}
  0< R &\leq & \varliminf_{n \to \infty} \frac{-1}{2n} \log c_k \nonumber\\
      &= &\varliminf_{n \to \infty} \frac{1}{2n} \log\frac{|{v}_{k,j_k}|^2}{c_k}\label{eq:pos}\\
            &= &\varliminf_{n \to \infty} \frac{1}{2n} \log\left(1+\frac{|{v}_{k,j_k}|^2}{c_k}\right) \label{rate111}
\end{IEEEeqnarray}
where the first equality  holds by Condition~\eqref{Lemma1C} and the second equality holds because~\eqref{eq:pos} implies that the ratio $\frac{|{v}_{k,j_k}|^2}{c_k}$ tends to infinity with $n$.

Combining~\eqref{rate11} with~\eqref{rate111} establishes that for arbitrary $\epsilon>0$ there exists a rate tuple $(R_1,\ldots, R_K)$ satisfying \eqref{eq:rate} such that  the described scheme with independent messages  achieves probability of error that tends to 0 as the blocklength tends to infinity. 

\begin{remark}\label{rem:region}
In the spirit of the scheme for private messages described above, one can construct a linear-feedback scheme with a common message point that has arbitrary small probability of error whenever
\begin{equation*}
R< \varlimsup_{n\to \infty} -\frac{1}{2n} \log c_k, \qquad k\in\set{K}.
\end{equation*}
Combined with the converse in~\eqref{gamma10}, this gives a (multi-letter) characterization of $C^{\textnormal{(Lin)}}$. Based on this multi-letter characterization one can show the continuity of $C^{\textnormal{(Lin)}}$ in the transmit-power constraint $P$.
\end{remark}
%

\section{Proof of Theorem~\ref{thm:2}: Coding Scheme achieving $L$-th order exponential decay}\label{sec:theorem2}

The scheme is based on the scheme in~\cite{reza}, see also~\cite{stark}.
Fix a  positive rate $R<C$ and a positive integer $L$. Assume that
\begin{equation}\label{eq:fb}
R_\textnormal{fb}\geq R(L-1).
\end{equation}
Also, fix a large blocklength $n$ and small numbers $\epsilon, \delta > 0$  such that
\begin{equation}\label{eq:RC}
 R<C(1-\delta)
\end{equation}
and
\begin{IEEEeqnarray}{rCl}
(1-\epsilon)^{-1} &<&1+\delta  \label{eq:epsdelta} .
\end{IEEEeqnarray}
Define
\begin{equation}\label{eq:n2}
n_1:=(1-\epsilon) n
\end{equation}
 and for $l\in\{2,\ldots,L\}$
 \begin{equation}\label{eq:n1}
 n_l:=n_1+\frac{\epsilon n}{L-1}(l-1).
 \end{equation}
Notice that by~\eqref{eq:epsdelta} and \eqref{eq:n2},
\begin{equation}\label{eq:delta}
\frac{n}{n_1}  <1+\delta.
\end{equation}
The coding scheme takes place in $L$ phases. After each phase $l\in\{1,\ldots,L\}$, each Receiver $k\in\set{K}$ makes a temporary guess $\hat{M}^{(k)}_l$ of message $M$.
The final guess is the guess after phase~$L$:
\begin{IEEEeqnarray}{rCl}
 \hat{M}^{(k)}=\hat{M}_L^{(k)},
\end{IEEEeqnarray}
Define the probability of error after phase $l\in\{1,\ldots,L\}$:
\begin{IEEEeqnarray}{rCl}\label{eq:Pel}
  P^{(n)}_{e,l} :=\Prv{\bigcup_{k\in\set{K}} \hat{M}^{(k)}_l\neq M}
\end{IEEEeqnarray}
and thus
\begin{equation}\label{eq:toterror}
  P^{(n)}_{e}=P^{(n)}_{e,L}.
\end{equation}


\subsection{Code Construction}
We construct a codebook $\mathcal{C}_1$ that
\begin{itemize}
\item is of blocklength $n_1$,
\item is of rate $R_{\text{phase},1}=\frac{n}{n_1}R$,
\item satisfies an expected average block-power constraint $P$, and
\item when used to send a common message over the Gaussian BC in~\eqref{eq:BCk} and combined with an optimal decoding rule,  it achieves probability of error $\rho_1$ not exceeding
\begin{IEEEeqnarray}{rCl}\label{eq:rho1}
\rho_1\leq e^{ -n(\zeta -o(1))}
\end{IEEEeqnarray}
for some  $\zeta>0$. 
\end{itemize}Notice that such a code exists because, by \eqref{eq:RC} and \eqref{eq:delta}, the rate of the code $\frac{n}{n_1}R< C (1-\delta^2)$, and because the error exponent of the BC with common message without feedback is positive for all rates below capacity.\footnote{The positiveness of the error exponent for the Gaussian BC with common message and without feedback follows from the fact that without feedback the probability of error for the Gaussian BC with common messages is at most $K$ times the probability of error to the weakest receiver.}

Let
\begin{equation}\label{gamma1}
\gamma_1:=\rho_1.
\end{equation}
For $l$  from 2 to $L$, do the following.

Construct a codebook $\mathcal{C}_l$ that:
\begin{itemize}
\item is of blocklength $\frac{\epsilon n}{L-1}-1$,
\item is of rate $R_{\text{phase},l}:=\frac{R(L-1)}{\epsilon-(L-1)/n}$,
\item satisfies an expected average block-power constraint $P/{\gamma_{l-1}}$,
\item when used to send a common message over the Gaussian BC in~\eqref{eq:BCk} and combined with an optimal decoding rule,  it achieves probability of error $\rho_l$ not exceeding
\begin{equation}\label{eq:rho}
\rho_l\leq \textnormal{exp}(-\underbrace{\textnormal{exp}\circ\ldots\circ\textnormal{exp}}_{l-1~\text{times}}(\Omega(n))).
\end{equation}
\end{itemize}
Define
\begin{IEEEeqnarray}{rCl}\label{eq:gamrho}
{\gamma_l}:= \rho_l+2\sum_{k\in\set{K}}\set{Q}\left( \frac{\sqrt{P/\gamma_{l-1}}}{2\sigma_k}\right).
\end{IEEEeqnarray}
(As shown in Section~\ref{subsec:analysis} ahead, $\gamma_l$  upper bounds   $P^{(n)}_{e,l}$ defined in~\eqref{eq:Pel}.)
By  \eqref{eq:rho} and~\eqref{eq:gamrho}, inductively one can show that
\begin{IEEEeqnarray}{rCl}\label{eq:expell}
{\gamma_l}\leq \textnormal{exp}(-\underbrace{\textnormal{exp}\circ\ldots\circ\textnormal{exp}}_{l-1~\text{times}}(\Omega(n))).
\end{IEEEeqnarray}


In Appendix~\ref{app:ex}, we prove that such codes $\set{C}_2, \ldots, \set{C}_L$ exist.

\subsection{Transmission}
Transmission takes place in $L$ phases.
\subsubsection{\textit{First phase with channel uses $i=1,\ldots, n_1$}} During the first $n_1$ channel uses, the transmitter sends the codeword in $\set{C}_1$ corresponding to message $M$.

After observing the channel outputs $Y^{n_1}_{k}$, Receiver~$k\in\set{K}$ makes a temporary decision $\hat{M}^{(k)}_{1}$ about $M$. It then sends this temporary decision $\hat{M}^{(k)}_{1}$  to the transmitter over the feedback channel:
\begin{equation}\label{eq:fbPha1}
V_{k, n_1} = \hat{M}^{(k)}_{1}.
\end{equation}
All previous feedback signals from Receiver~$k$ are deterministically 0. 

\subsubsection{ Phase $l\in\{2,\ldots,L\}$ with channel uses $i\in\{n_{l-1}\!+\!1, \ldots, n_l\}$}
The communication in phase $l$ depends on the receivers' temporary decisions $\hat{M}^{(1)}_{l-1}, \ldots, \hat{M}^{(K)}_{l-1}$ after the previous phase $(l-1)$. These decisions have been communicated to the transmitter over the respective feedback links.

If in phase $(l-1)$ at least one of the receivers made an incorrect decision,
\begin{equation}\label{eq:somewrong}
 (\hat{M}^{(k)}_{l-1} \neq M) , \qquad \textnormal{for some } k\in\set{K},
\end{equation}
then in channel use $n_{l-1}+1$ the transmitter sends an error signal to indicate an error:
\begin{IEEEeqnarray}{rCl}
X_{n_l+1} = \sqrt{P/{\gamma_{l-1}}}.
\end{IEEEeqnarray}
During the remaining channel uses $i=n_{l-1}+2, \ldots, n_l$ it then retransmits the message~$M$ by sending the codeword from $\set{C}_l$ that corresponds to $M$.

On the other hand, if all receivers' temporary decisions to the phase $(l-1)$ were correct,
\begin{equation}\label{eq:allcorrect}
 \hat{M}^{(1)}_{l-1}= \hat{M}^{(2)}_{l-1}=\ldots = \hat{M}^{(K)}_{l-1}= M,
\end{equation}
then the transmitter sends 0 during the entire phase $l$:
\begin{equation}
X_i= 0, \qquad i=n_{l-1}+1, \ldots, n_l.
\end{equation}
In this case, no power is consumed in phase $l$.

The receivers first detect whether the transmitter sent an error signal in channel use $n_{l-1}+1$. Depending on the output of this detection, they  either stick to their temporary decision in phase $(l-1)$ or  make a new decision based on the transmissions in phase $l$.  Specifically, if
\begin{equation}
Y_{k,n_{l-1}+1}< T_{l-1}
\end{equation}
 where
\begin{equation}\label{eq:GAMMA}
T_{l-1} := \frac{\sqrt{P/{\gamma_{l-1}}}}{2},
\end{equation}
then  Receiver~$k\in\set{K}$ decides that its decision $\hat{M}_{l-1}^{(k)}$ in phase $(l-1)$  was correct and keeps it as its temporary guess of the message $M$:
\begin{equation}
 \hat{M}_l^{(k)}= \hat{M}_{l-1}^{(k)}.
\end{equation}
 If instead,
 \begin{equation}
Y_{k,n_{l-1}+1}\geq T_{l-1},
\end{equation}
Receiver~$k$ decides that its temporary decision $\hat{M}_{l-1}^{(k)}$ was wrong and  discards it. It then produces a new guess $\hat{M}_{l}^{(k)}$ by decoding the code $\set{C}_l$ based on the outputs $Y_{k,n_{l-1}+2}, \ldots, Y_{k,n_l}$.

After each phase $l\in\{2,\ldots,L-1\}$, each Receiver $k\in\set{K}$ feeds back to the transmitter its temporary guess $\hat{M}_{l}^{(k)}$:
\begin{equation}\label{eq:fbPhal}
V_{k, n_l} = \hat{M}^{(k)}_{l}.
\end{equation}
All other feedback signals $V_{k, n_{l\!-\!1}\!+\!1}\!,\ldots,\!V_{k,n_l\!-\!1}$ in phase $l$  are deterministically 0. 

After $L$ transmission phases, Receiver $k$'s final guess is
\begin{equation}
 \hat{M}^{(k)}= \hat{M}_{L}^{(k)}.
\end{equation}
Thus, an error occurs in the communication if 
\begin{equation}
(\hat{M}^{(k)}_L\neq M),~~\text{for~some}~k\in\set{K}.
\end{equation}

\subsection{Analysis}\label{subsec:analysis}
In view of~\eqref{eq:fb}, by~\eqref{eq:fbPha1} and~\eqref{eq:fbPhal}, and because all other feedback signals are deterministically 0, our scheme satisfies the feedback rate  constraint in~\eqref{eq:fbconstra}.

We next analyze the probability of error and we bound the consumed power.
These analysis rely on the following events. For each $k\in\set{K}$ and $l\in\{1,\ldots,L\}$ define the events:
\begin{itemize}
\item $\set{E}^{(k)}_{l}$: Receiver~$k$'s decision in phase $l$ is wrong:
\begin{equation}
\hat{M}_{l}^{(k)} \neq M;
\end{equation}
\item $\set{E}^{(k)}_{T,l}$: Receiver~$k$ observes
\begin{equation}
Y_{k, n_{l}+1}<T_l;
\end{equation}
\item $\set{E}^{(k)}_{\rho,l}$: Decoding Message $M$ based on Receiver~$k$'s phase-$l$ outputs $Y_{k,n_{l-1}+2},\ldots,Y_{k,n_{l}}$ using codebook $\set{C}_{l}$ results in an error. 
\end{itemize}
Define also the events:
\begin{itemize}
\item[$\set{E}_{1,l}$:] All receivers' decisions in  phase $(l-1)$ are correct, and at least one Receiver $k\in\set{K}$ obtains an error signal in channel use $n_{l-1}+1$ :
\begin{equation} \label{eq:E1}
 \bigg(\bigcap_{k\in\set{K}} \big(\set{E}^{(k)}_{l-1}\big)^c\bigg)
\cap \bigg(
\bigcup_{k\in\set{K}} \big(\set{E}_{T,l-1}^{(k)}\big)^c\bigg).
 \end{equation}
\item[$\set{E}_{2,l}$:] At least one Receiver $k\in\set{K}$ makes an incorrect decision in phase $(l-1)$ but obtains no error signal in channel use $n_{l-1}+1$:
\begin{equation}\label{eq:E2}
\bigcup_{k\in\set{K}} \bigg( \set{E}^{(k)}_{l-1} \cap \set{E}^{(k)}_{T,l-1}\bigg).
\end{equation}
\item[$\set{E}_{3,l}$:] 
At least one Receiver~$k\in\set{K}$ makes an incorrect temporary decision in  phase $(l-1)$, and at least one Receiver~$k'\in\set{K}$ observes $Y_{k', n_{l-1}+1}\geq T_{l-1}$ and errs when decoding $M$ based on its phase-$l$ outputs $Y_{k',n_{l-1}+2},\ldots,Y_{k',n_{l}}$: 
\begin{equation}\label{eq:E3}
\bigg( \bigcup_{k\in \set{K}} \set{E}^{(k)}_{l-1} \bigg) \cap \bigg( \bigcup_{k'\in \set{K}} \Big(\big(\set{E}^{(k')}_{T,l} \big)^c \cap \set{E}^{(k')}_{\rho,l}\Big) \bigg).
\end{equation}
\end{itemize}

For each $l\in\{1,\ldots,L\}$, the probability $P_{e,l}^{(n)}$ is included in the union of the events $(\set{E}_{1,l}\cup \set{E}_{2,l}\cup \set{E}_{3,l})$, and thus, by the union bound,
\begin{IEEEeqnarray}{rCl}\label{eq:errsum}
  P_{e,l}^{(n)}  & \leq &
  \Prv{\set{E}_{1,l}} + \Prv{\set{E}_{2,l}} + \Prv{\set{E}_{3,l}}.
\end{IEEEeqnarray}
In particular, by~\eqref{eq:toterror} and \eqref{eq:errsum}, the probability of error of our scheme
\begin{IEEEeqnarray}{rCl}\
  P_e^{(n)} &\leq&
  \Prv{\set{E}_{1,L}} + \Prv{\set{E}_{2,L}} + \Prv{\set{E}_{3,L}}.\label{eq:Esum}
\end{IEEEeqnarray}


We bound each summand in~\eqref{eq:Esum} individually, starting with $\Prv{\set{E}_{1,L}}$.
By~\eqref{eq:E1}, we have
\begin{IEEEeqnarray}{rCl}
\Prv{\set{E}_{1,L}} & = & \Prv{ \bigg(\bigcap_{k\in\set{K}} \big(\set{E}^{(k)}_{L-1}\big)^c\bigg)
\cap \bigg(
\bigcup_{k\in\set{K}} \big(\set{E}_{T,L-1}^{(k)}\big)^c\bigg)} \nonumber \\
& \leq &   \Prv{
\bigcup_{k\in\set{K}} \big(\set{E}_{T,L-1}^{(k)}\big)^c \Big| \bigcap_{k\in\set{K}} \big(\set{E}^{(k)}_{L-1}\big)^c  }\nonumber\\
& \leq  &   \sum_{k=1}^K \Prv{  \big(\set{E}_{T,L-1}^{(k)}\big)^c   \big|\bigcap_{k\in\set{K}} \big(\set{E}^{(k)}_{L-1}\big)^c }\nonumber \\
& = & \sum_{k=1}^K \set{Q}\bigg(\frac{T_{L-1}}{\sigma_k}\bigg)\label{eq:E1bound}
\end{IEEEeqnarray}
where the first inequality follows by Bayes' rule and because a probability cannot exceed 1; the second inequality by the union bound; and the last equality because in the event $\big(\bigcap_{k\in\set{K}} (\set{E}^{(k)}_{L-1})^c\big)$, we have $X_{n_{L-1}+1}=0$ and thus  $Y_{k, n_{L-1}+1} \sim \mathcal{N}(0,\sigma^2_k)$.

Next, by~\eqref{eq:E2} and similar arguments as before, we obtain,
\begin{IEEEeqnarray}{rCl}
\Prv{\set{E}_{2,L}} &=& \Prv{ \bigcup_{k=\in \set{K}} \big(\set{E}^{(k)}_{L-1} \cap \set{E}^{(k)}_{T,L-1}\big)} \nonumber\\
& \leq & \sum_{k=1}^{K} \Prv{  \set{E}^{(k)}_{L-1}  \cap \set{E}^{(k)}_{T,L-1}}\nonumber \\
& \leq & \sum_{k=1}^{K} \Prv{  \set{E}^{(k)}_{T,L-1}\big| \set{E}^{(k)}_{L-1} }\nonumber \\
& = & \sum_{k=1}^K \set{Q}\bigg(\frac{T_{L-1}}{\sigma_k}\bigg). \label{eq:E2bound}
\end{IEEEeqnarray}
Finally, by~\eqref{eq:E3} and similar arguments as before,
\begin{IEEEeqnarray}{rCl}
\Prv{\set{E}_{3,L}} & = & \Prv{ \! \bigg( \bigcup_{k\in \set{K}} \set{E}^{(k)}_{L-1} \bigg) \cap \bigg( \bigcup_{k'\in \set{K}} \Big(\big(\set{E}^{(k')}_{T,L} \big)^c \cap \set{E}^{(k')}_{\rho,L}\Big) \bigg) } \nonumber \\
& \leq & \Prv{ \bigcup_{k'\in \set{K}} \Big(\big(\set{E}^{(k')}_{T,L} \big)^c \cap \set{E}^{(k')}_{\rho,L}\Big)
 \Big |  \bigcup_{k\in \set{K}} \set{E}^{(k)}_{L-1} } \nonumber \\
& \leq  & \Prv{ \bigcup_{k'\in \set{K}} \set{E}^{(k')}_{\rho,L}
 \Big |  \bigcup_{k\in \set{K}} \set{E}^{(k)}_{L-1} } \nonumber \\
&\leq& \rho_L \label{eq:E3bound}
\end{IEEEeqnarray}
where the last inequality follows by the definition of $\rho_L$.

In view of (\ref{eq:GAMMA}) and (\ref{eq:Esum})--
(\ref{eq:E3bound}),
\begin{IEEEeqnarray}{rCl}\label{eq:PeGa}
   P_e^{(n)} &\leq& \Prv{\set{E}_{1,L}} + \Prv{\set{E}_{2,L}} + \Prv{\set{E}_{3,L}} \nonumber\\
    &\leq& \rho_L+2\sum_{k\in\set{K}}\set{Q}\left( \frac{\sqrt{P/\gamma_{L-1}}}{2\sigma_k}\right)\nonumber\\
            &=& \gamma_L
\end{IEEEeqnarray}
where the equality follows by the definition of $\gamma_L$ in (\ref{eq:gamrho}).
Combining this with the $L$-th order exponential decay of $\gamma_L$, see~\eqref{eq:expell}, we obtain
 \begin{equation}
\varliminf_{n\to \infty} -\frac{1}{n} \underbrace{\log  \log \ldots \log}_{L-1 \,\textnormal{times}} (-\log  P_{e}^{(n)} )>0,
 \end{equation}

Now consider the consumed expected average block-power. Similarly to~\eqref{eq:PeGa}, we can show that for $l\in\{1,\ldots,L-1\}$,
\begin{IEEEeqnarray}{rCL}\label{eq:PeGaPh}
  P_{e,l}^{(n)} &\leq& 
           \gamma_l.
\end{IEEEeqnarray}
Since in each phase $l\in\{2,\ldots,L\}$ we consume power $P/\gamma_{l-1}$ in the event (\ref{eq:somewrong}) and power 0 in the event (\ref{eq:allcorrect}), by the definition in~\eqref{eq:Pel},
\begin{equation}
\frac{1}{n} \E{\sum_{i=1}^n X_i^2}\leq \frac{1}{n}\!\Big(P(1\!-\!\epsilon)n\!+\!\sum_{l=2}^{L}\!P^{(n)}_{e,l\!-\!1}\frac{P}{\gamma_{l-1}}\frac{\epsilon n}{L\!-\!1}\Big)\leq P
\end{equation}
where the second inequality follows from (\ref{eq:PeGaPh}).

This completes the proof of Theorem~\ref{thm:2}.



\appendices \label{prlema1}

\section{Proof of Proposition~\ref{prop:condnoi}}\label{app1}
We show that under assumption~\eqref{eq:condnoisevariance},
\begin{equation}
\lim_{K\to \infty} \alpha_1^\star =0,
\end{equation}
which implies~\eqref{eq:czero}.

Notice that~\eqref{eq:Requal} implies
for $k\in\{1\!,\ldots,\! K\!-\!1\}$:
\begin{eqnarray}
\frac{\alpha^*_K P}{N_K}=\frac{\alpha^*_k P}{(1-\alpha_1^\star-\alpha^*_2....-\alpha^*_k)P+N_k}.
\end{eqnarray}
Since for each $k$, the term $(1- \alpha_1^\star- \alpha^*_2- \ldots- \alpha^*_k)$ is nonnegative,
 \begin{eqnarray}
  \alpha^*_k& \geq	& \frac{N_k}{N_K}\alpha^*_K, \qquad k\in\{1,\ldots, K-1\}.
\end{eqnarray}
Thus, by~\eqref{eq:alphasum},
\[
1=  \sum_{k=1}^K \alpha^*_k \geq \sum_{k=1}^K \frac{N_k}{N_K} \alpha^*_K
\]
and
\[\alpha^*_K\leq \frac{N_K}{\sum_{k=1}^K N_k}.
\]
We conclude that, for every finite positive integer $K$,
\begin{eqnarray*}
R_K &\leq& \frac{1}{2}\log\bigg(1+\frac{P}{\sum_{k=1}^K N_k}\bigg),
\end{eqnarray*}
and under Assumption~\eqref{eq:condnoisevariance}, in the limit as $K\to \infty$,
\[
\lim_{K\to \infty} R_K =0.
\]

\section{Proof of Lemma~\ref{lem:vlog}}\label{app:fano}
We first prove the converse~\eqref{gamma10}.
Fix a blocklength $n$.
By Fano's inequality, for each $k\in\set{K}$,
\begin{IEEEeqnarray}{rCl} \label{fano1}
nR &=&H(M^{(n)}) \nonumber \\
     &\leq& I\left(M^{(n)}; Y_{k,1}^{(n)}, \ldots,  Y_{k,n}^{(n)}\right)+\epsilon(n)\nonumber\\
     &\leq &I\left(\Theta^{(n)}; Y_{k,1}^{(n)}, \ldots,  Y_{k,n}^{(n)}\right)+\epsilon(n)\nonumber\\
     &\stackrel{(a)} \leq &I\left(\bar{\Theta}^{(n)};\bar{Y}_{k,1}^{(n)}, \ldots, \bar Y_{k,n}^{(n)}\right)+\epsilon(n)
\end{IEEEeqnarray}
where $\frac{\epsilon(n)}{n} \to 0$ as $n\to \infty$ and where we defined the tuple $(\bar{\Theta}^{(n)}, \bar{Y}_{k,1}^{(n)}, \ldots, \bar Y_{k,n}^{(n)})$ to be jointly Gaussian with the same covariance matrix as the tuple $(\Theta^{(n)}; Y_{k,1}^{(n)}, \ldots,  Y_{k,n}^{(n)})$. Inequality (a) holds because the Gaussian distribution maximizes differential entropy under a covariance constraint.

Now, since $\bar{\Theta}^{(n)}, \bar{Y}_{k,1}^{(n)}, \ldots, \bar Y_{k,n}^{(n)}$ are jointly Gaussian, there exists a linear combination $\sum_{i=1}^n v_{k,i}^{(n)}\bar Y_{k,i}^{(n)}$ such that
\begin{equation}\label{eq:view}
I\left(\bar{\Theta}^{(n)}; \bar{Y}_{k,1}^{(n)}, \ldots, \bar Y_{k,n}^{(n)}\right)=I\bigg(\bar{\Theta}^{(n)};\sum_{i=1}^n v_{k,i}^{(n)} \bar Y_{k,i}^{(n)}\bigg).
\end{equation}
(In fact, the linear combination is simply the LMMSE-estimate of $\bar \Theta^{(n)}$ based on $\bar{Y}_{k,1}^{(n)}, \ldots, \bar Y_{k,n}^{(n)}$.) Defining the $n$-dimensional row-vector $\vect{v}_k^{(n)}= \big(v_{k,1}^{(n)}, \ldots, v_{k,n}^{(n)}\big)$,
in view of~\eqref{eq:view}, we have
\begin{IEEEeqnarray}{rCl}\label{eq:Ilog}
\lefteqn{
I\left(\bar{\Theta}^{(n)}; \bar{Y}_{k,1}^{(n)}, \ldots, \bar Y_{k,n}^{(n)}\right)} \quad \nonumber \\
& = & \frac{1}{2} \log\left( 1+ \frac{\big(\vect{v}_k^{(n)}\vect{d}^{(n)}\big)^2\Var{\Theta^{(n)}}}{c_k^{(n)}}
\right)
\end{IEEEeqnarray}
where $c_k^{(n)}$ is as defined in~\eqref{eq:cu}.

Notice that the right-hand side of~\eqref{eq:Ilog} does not depend on the norm of $\vect{v}_k^{(n)}$ (as long as it is non-zero) but only on the direction. Therefore, without loss of generality, we can assume that
\begin{equation}\label{eq:unitnormv}
\|\vect{v}_k^{(n)}\|^2=1.
\end{equation}
By~\eqref{fano1} and~\eqref{eq:Ilog}, we  conclude that for each $k\in\set{K}$, there exists a unit-norm vector $\vect{v}_k^{(n)}$ such that
\begin{IEEEeqnarray}{rCl}\label{gamma13}
R&\leq& \varliminf_{n \to \infty} \frac{1}{2n} \log\left(\! 1\!+\! \frac{\big(\vect{v}_k^{(n)}\vect{d}^{(n)}\big)^2\Var{\Theta^{(n)}}}{c_k^{(n)}}\!\right).
\end{IEEEeqnarray}
Since by assumption $R>0$, \eqref{gamma13} implies that the ratio  $(\vect{v}_k^{(n)}\vect{d}^{(n)})^2\Var{\Theta^{(n)}}/c_k^{(n)}$ tends to infinity and thus
\begin{IEEEeqnarray}{rCl}\label{gamma100}
R&\leq& \varliminf_{n \to \infty} \frac{1}{2n} \log\left(  \frac{\big(\vect{v}_k^{(n)}\vect{d}^{(n)}\big)^2\Var{\Theta^{(n)}}}{c_k^{(n)}}
\right).
\end{IEEEeqnarray}
 Now, consider the average block-power constraint~\eqref{eq:dd}. Since the trace of a positive semidefinite matrix is non-negative and $\Var{\Theta^{(n)}}\leq \E{ |\Theta^{(n)}|^2}$, by~\eqref{eq:dd}, for each $n\in\Integers^{+}$:
\begin{equation}\label{eq:dpower}
\|\vect{d}^{(n)}\|^2 \E{|\Theta^{(n)}|^2}\leq n(P-\delta).
\end{equation}
Since $\|\vect{v}_k^{(n)}\|=1$, \eqref{eq:unitnormv},   by the Cauchy-Schwarz Inequality,
\begin{equation}
\big(\vect{v}_k^{(n)}\vect{d}^{(n)}\big)^2\Var{\Theta^{(n)}} \leq n(P-\delta)
\end{equation}
and as a consequence
\begin{equation}\label{eq:numneg}
\varliminf_{n\to \infty}
 \frac{1}{2n} \log\left(\big(\vect{v}_k^{(n)}\vect{d}^{(n)}\big)^2\Var{\Theta^{(n)}}\right) \leq 0.
 \end{equation}
 Combining this with \eqref{gamma100}, proves the desired inequality~\eqref{gamma10}.

The proof of Inequalities~\eqref{Lemma1A} and \eqref{Lemma1C} relies on Lemmas~\ref{Lemma5} and \ref{Lemma3} at the end of this appendix.
Notice that the monotonicity of the $\log$-function and the nonnegativity of the norm combined with \eqref{gamma10} imply that for each $k\in\set{K}$,
 \begin{IEEEeqnarray}{rCl}
R & \leq &  \varliminf_{n\to \infty} - \frac{1}{2n} \log \big\|\vect{v}_k^{(n)} \big(\mat{I}+\mat{A}_k^{(n)}\big)\big\|^2, \label{eq:Gaminequa1}
 \end{IEEEeqnarray}
 where recall that we assumed $R>0$.

Define  for each $k\in\set{K}$ and positive integer $n$ the set
\begin{IEEEeqnarray}{rCl}\label{eq:S1}
\set{S}_k^{(n)} &: = &\left\{ i \in\{1,\ldots, n\} \colon v_{k,i}^{(n)} > {n^{-2\log n}}\right\}.
\end{IEEEeqnarray}
By Lemma~\ref{Lemma5} and Inequality~\eqref{eq:Gaminequa1}, the cardinality of  each set $\set{S}_k^{(n)}$ is unbounded,
\begin{IEEEeqnarray}{rCl}
|\set{S}_k^{(n)}| \to \infty \qquad \textnormal{as} \qquad n\to \infty, \qquad k\in\set{K}.
\end{IEEEeqnarray}
Applying now Lemma~\ref{Lemma3}  to $p=P-\delta$, to
\begin{equation}
\pi_{i}^{(n)} = \E{\left(X_{i}^{(n)}\right)^2},
\end{equation}
and to $\set{T}^{(n)}= \set{S}_k^{(n)}$ implies that for each $k\in\set{K}$ there exists a sequences of indices $\{j_k^{(n)}\in\set{S}_k^{(n)}\}_{n=1}^{\infty}$ that satisfies \eqref{Lemma1A}. Since every sequence of indices $\{i^{(n)}  \in \set{S}_k^{(n)}\}_{ n=1}^{\infty}$ also satisfies \eqref{Lemma1C}, this concludes the proof of the lemma.

 \begin{Lemma}\label{Lemma5}
 For each $n\in \Integers^+$, let $\mat{A}^{(n)}$ be a strictly lower-triangular $n$-by-$n$ matrix and  $\vect{v}^{(n)}$  an $n$-dimensional row-vector. Let $a_{i,j}^{(n)}$ denote the row-$i$, column-$j$ entry of $\mat{A}^{(n)}$ and $v_{i}^{(n)}$ denote the $i$-th entry of $\vect{v}^{(n)}$.
  Assume that the elements $a_{i,j}^{(n)}$ are bounded as
 \begin{equation}\label{eq:powerc}
| a_{i,j}^{(n)}|^2 \leq n p
 \end{equation}
 for some  real number $p>0$,
 and that the inequality
 \begin{equation}\label{eq:Gamlim}
\varliminf_{n\to \infty} - \frac{1}{2n}\log \|\vect{v}^{(n)} (\mat{I}+ \mat{A}^{(n)})\|^2\geq\Gamma
 \end{equation}
holds for some real number $\Gamma>0$.
Then, for each $\epsilon\in (0,\Gamma)$ and for all sufficiently large $n$ the following implication holds: If
\begin{subequations}\label{eq:impli}
 \begin{equation}\label{eq:v1j}
|v^{(n)}_{j}|>e^{-n(\Gamma-\epsilon)}
 \end{equation}
 for some index $j\in\{1,\ldots, n\}$,
then there must exist an index $i\in\{j+1,\ldots, n\}$ such that
  \begin{equation}\label{eq:lempr}
  |v^{(n)}_{i}|\geq \frac{|v^{(n)}_{j}|-e^{-n(\Gamma-\epsilon)}}{n^{\frac{3}{2}}\sqrt{p}}.
    \end{equation}
  \end{subequations}

If moreover, the vectors $\{\vect{v}^{(n)}\}_{n=1}^{\infty}$ are of unit norm, then
the cardinality of the set
\begin{equation}\label{eq:setn}
\mathcal{S}^{(n)}:=\left\{j \in \{1,\ldots, n\}:\, |v_{j}^{(n)}|>n^{-2\log(n)}\right\}
\end{equation}
is unbounded in $n$.
\end{Lemma}
\begin{proof}
Fix $\epsilon\in(0,\Gamma)$ and let $n$ be sufficiently large so that
  \begin{equation}\label{eq:Gammaeps}
  -\frac{1}{2n}\log\big\|\vect{v}^{(n)} (\mat I+\mat{A}^{(n)})\big\|^2\geq \Gamma-\epsilon.
  \end{equation}
  This is possible by~\eqref{eq:Gamlim}.

Since $\mat{A}^{(n)}$ is strictly lower-triangular,
  \[
  \begin{array}{ll}
    \|\vect{v}^{(n)}(\mat I+ \mat{A}^{(n)})\|^2
    \! \!&=\sum\limits_{j=1}^{n}({v}^{(n)}_{j}+\sum\limits_{i=j+1}^{n}{v}^{(n)}_{i}a_{i,j}^{(n)})^2\\
    & \geq \sum\limits_{j=1}^{n}\big(|{v}^{(n)}_{j}|-\big|\sum\limits_{i=j+1}^{n}{v}^{(n)}_{i}a_{i,j}^{(n)}\big|\big)^2\\
    &\geq \big(|{v}^{(n)}_{j}|-\big|\sum\limits_{i=j+1}^{n}{v}^{(n)}_{i}a_{i,j}^{(n)}\big|\big)^2
  \end{array}
  \]
and  by~\eqref{eq:Gammaeps} and the monotonicity of the $\log$-function,   for 
all $j\in\{1,\ldots, n\}$:
  \[
    -\frac{1}{2n}\log\bigg(|{v}^{(n)}_{j}|-\big|\sum\limits_{i=j+1}^{n}v^{(n)}_{i}a_{i,j}^{(n)}\big|\bigg)^2\geq \Gamma-\epsilon.
  \]
  Thus,
  \[
  |v^{(n)}_{j}|\leq \big|\sum\limits_{i=j+1}^{n}v^{(n)}_{i}a_{i,j}^{(n)}\big|+ e^{-n(\Gamma-\epsilon)}
  \]
  and by~\eqref{eq:powerc}:
  \begin{IEEEeqnarray}{rCl}
  |v^{(n)}_{j}|-e^{-n(\Gamma-\epsilon)} & \leq & \big|\sum\limits_{i=j+1}^{n}v^{(n)}_{i}a_{i,j}^{(n)}\big| \nonumber \\
  & \leq& \sum\limits_{i=j+1}^{n}|v^{(n)}_{i}|\sqrt{np}\label{eq:d}.
  \end{IEEEeqnarray}
  If   $|v^{(n)}_{j}|\leq e^{-n(\Gamma-\epsilon)}$, then the sum on the right-hand side of~\eqref{eq:d} can be empty, i.e., $j=n$. However, if
  \begin{equation}\label{eq:vjn}
  |v^{(n)}_{j}|>e^{-n(\Gamma-\epsilon)},
  \end{equation}
  then the sum needs to have at least one term. Indeed, if  (\ref{eq:vjn}) holds and $i<n$, there must exist an index $i \in\{j+1,...,n\}$ such that
  \begin{equation}
  \frac{1}{n}\Big(|v^{(n)}_{j}|-e^{-n(\Gamma-\epsilon)}\Big)\leq |v^{(n)}_{i}|\sqrt{np},
  \end{equation}
  which is equivalent to the desired bound~\eqref{eq:lempr}.

We now prove the second part of the lemma, i.e., the unboundedness of the cardinalities of the sets $\set{S}^{(n)}$, where we assume that the vectors $\{\vect{v}^{(n)}\}$ are of unit norm.  In the following, let $n$ be sufficiently large so that the first part of the lemma, Implication~\eqref{eq:impli}, holds and so that
\begin{equation}\label{eq:nsuff0}
\frac{1}{\sqrt{n}} > \frac{1}{n^{2\log(n)}} >e^{-n(\Gamma-\epsilon)}
\end{equation}
and
for every $\ell\in\{1,\ldots, \log(n)\}$
  \begin{IEEEeqnarray}{rCl}
 \lefteqn{  \frac{1}{n^{(3\ell+1)/2}p^{\ell/2}} - e^{-n(\Gamma-\epsilon)} n^{-3/2}p^{-1/2} \frac{1- n^{-3 \ell/2} p^{-\ell/2}}{1- n^{-3/2}p^{-1/2}} } \quad \nonumber \\
  &>& \frac{1}{n^{2\log(n)}} \hspace{6cm}\label{eq:nlogsuff}
  \end{IEEEeqnarray}

Since $\|\vect{v}^{(n)}\|^2=1$, for each $n$,  there must exist an index $i_0^{(n)}\in\{1,\ldots, n\}$  such that
\begin{equation}\label{eq:average}
|v^{(n)}_{i_0^{(n)}}| \geq \frac{1}{\sqrt{n}},
\end{equation}
and  by \eqref{eq:nsuff0}
\begin{equation}
|v^{(n)}_{i_0^{(n)}}|> n^{2\log(n)} >e^{-n(\Gamma-\epsilon)}.
\end{equation}
We conclude by~\eqref{eq:impli} that there exists an index $i_1^{(n)}\in \{i_0^{(n)}+1, \ldots, n\}$ satisfying
  \begin{IEEEeqnarray}{rCl}
    |v^{(n)}_{i_{1}^{(n)}}|&\geq& \frac{|v^{(n)}_{i_0^{(n)}}|-e^{-n(\Gamma-\epsilon)}}{n^{\frac{3}{2}}\sqrt{p}}\nonumber \\
    & \geq & \frac{1}{n^2 \sqrt{p}}-\frac{e^{-n(\Gamma-\epsilon)}}{n^{\frac{3}{2}}\sqrt{p}}\label{eq:vi1}
  \end{IEEEeqnarray}
  where the  inequality follows from \eqref{eq:average}.
  By~\eqref{eq:nsuff0} and~\eqref{eq:nlogsuff}, (applied for $\ell=1$), Inequality~\eqref{eq:vi1} implies that
   \begin{equation}
|v^{(n)}_{i_1^{(n)}}|>e^{-n(\Gamma-\epsilon)},
\end{equation}
and consequently, by~\eqref{eq:impli}, there exists an index $i_2^{(n)}\in \{i_1^{(n)}+1, \ldots, n\}$ satisfying
  \begin{IEEEeqnarray}{rCl}\label{eq:vi2}
    |v^{(n)}_{i_2^{(n)}}|& \geq &\frac{|v^{(n)}_{i_1^{(n)}}|-e^{-n(\Gamma-\epsilon)}}{n^{\frac{3}{2}}\sqrt{p}}\\
  & \geq & \frac{1}{ n^{7/2} p} -\frac{e^{-n(\Gamma-\epsilon)}}{n^{3}p}-\frac{e^{-n(\Gamma-\epsilon)}}{n^{\frac{3}{2}}\sqrt{p}}\\
  & >&e^{-n(\Gamma-\epsilon)},
  \end{IEEEeqnarray}
  where the last inequality follows by~\eqref{eq:nsuff0} and~\eqref{eq:nlogsuff} (applied for $\ell=2$).

  Repeating these arguments iteratively, we  conclude that it is possible to find indices $1\leq i_0^{(n)}< i_1^{(n)} < \ldots < i_{\log(n)}^{(n)}< n$ such that  for each $\ell\in\{1,\ldots, \log(n)\}$:
  \begin{IEEEeqnarray}{rCl}
  | v_{i_\ell^{(n)}}^{(n)} |&\geq& \frac{1}{n^{3(\ell+1)/2}p^{\ell/2}} - e^{-n(\Gamma-\epsilon)} \sum_{j=1}^{\ell} \left(n^{-3/2}p^{-1/2}\right)^{j} \nonumber\\
  &= & \frac{1}{n^{3(\ell+1)/2}p^{\ell/2}}  \nonumber \\
  & & - e^{-n(\Gamma-\epsilon)} n^{-3/2}p^{-1/2} \frac{1- n^{-3 \ell/2} p^{-\ell/2}}{1- n^{-3/2}p^{-1/2}}\label{eq:vitera}\\
   &>& \frac{1}{n^{2\log(n)}}\\
   & > & e^{-(\Gamma-\epsilon)}
  \end{IEEEeqnarray}
 where the last two inequalities follow from~\eqref{eq:nsuff0} and~\eqref{eq:nlogsuff}.
This proves that for sufficiently large $n$ the cardinality of the set $\mathcal{S}^{(n)}$ as defined in~\eqref{eq:setn} is at least $\log(n)$ and thus unbounded in $n$. \end{proof}

 \begin{Lemma}\label{Lemma3}
 For each positive integer $n$, let $(\pi_1^{(n)},\ldots, \pi_n^{(n)})$ be a tuple of  nonnegative real numbers that satisfy
 \begin{equation}\label{eq:powerseq}
\frac{1}{n} \sum_{i=1}^n \pi_i^{(n)} \leq p
 \end{equation}
 for some real number $p>0$,
 and let $\set{T}^{(n)}$ be a subset of the indices from $1$ to $n$,
 \begin{equation}
 \set{T}^{(n)} \subseteq \{1,\ldots, n\},
 \end{equation}
 that satisfies
\begin{equation}
|\set{T}^{(n)} | \to \infty \quad \textnormal{ as } \quad n \to \infty.
\end{equation}
Then, there exists a sequence of indices $\Big\{i^{(n)}\in\set{T}^{(n)}\Big\}_{n=1}^{\infty}$ such that
\begin{equation}
\varlimsup_{n\to \infty}\frac{1}{n} \pi_{i^{(n)}}^{(n)} =0.
\end{equation}
\end{Lemma}
\begin{proof}
Since all numbers $\pi_{i}^{(n)}$ are nonnegative, for every sequence of indices $\{i^{(n)}\in\set{T}^{(n)}\}_{n=1}^{\infty}$,
\begin{equation}
\varlimsup_{n\to\infty} \frac{1}{n} \pi_{i^{(n)}}^{(n)} \geq 0.
\end{equation}
We thus have to prove that there exists at least one sequence of indices $\{i^{(n)}\in\set{T}^{(n)}\}_{n=1}^{\infty}$ that satisfies
\begin{equation}
\varlimsup_{n\to\infty} \frac{1}{n} \pi_{i^{(n)}}^{(n)} \leq 0.
\end{equation}
We prove this  by contradiction.
Assume  that for each sequence of indices $\{i^{(n)}\in \set{T}^{(n)}\}_{n=1}^{\infty}$
\begin{equation}\label{eq:ass}
\varlimsup_{n\to \infty}\frac{1}{n} \pi_{i^{(n)}}^{(n)} >0.
\end{equation}
Define for each $n\in\Integers^{+}$
\begin{equation}\label{eq:pmin}
\pi_{\min}^{(n)}:= \min_{i\in \set{T}^{(n)}} \pi_i^{(n)},
\end{equation}
and define the limit
\begin{equation}\label{eq:demin}
\delta_{\min}:=
\varlimsup_{n\to \infty}\frac{1}{n} \pi_{{\min}}^{(n)},
\end{equation}
which by Assumption~\eqref{eq:ass} is strictly positive,
\begin{equation}\label{eq:dpos}
\delta_{\min}>0.
\end{equation}

Now, since all the terms $\pi_{i}^{(n)}$ are nonnegative:
\begin{equation}\label{eq:sum}
\frac{1}{n}\sum_{i=1}^n \pi_{i}^{(n)} \geq\frac{1}{n} \sum_{i\in \set{T}^{(n)}} \pi_i^{(n)} \geq \frac{1}{n} \pi_{\min}^{(n)} |\set{T}^{(n)}|,
\end{equation}
where the second inequality follows by the definition in~\eqref{eq:pmin}.
By~\eqref{eq:demin} and \eqref{eq:dpos} and by the undboundedness of the cardinality of the sets  $\set{T}^{(n)}$, we conclude that  the sum in~\eqref{eq:sum} is unbounded in $n$, which contradicts Assumption~\eqref{eq:powerseq} and thus concludes our proof.
\end{proof}

\section{Existence of code $\set{C}_2,\ldots, \set{C}_L$ with the desired Properties}\label{app:ex}
The proof is by induction: for each $\ell\in\{2,\ldots, L\}$, when proving the existence of the desired $\set{C}_\ell$, we assume that
\begin{IEEEeqnarray}{rCl}\label{eq:gammalpr}
 \gamma_{l-1}\leq  \textnormal{exp}(-\underbrace{\textnormal{exp}\circ\ldots\circ\textnormal{exp}}_{l-2~\text{times}}(\Omega(n))).
\end{IEEEeqnarray}
For $l=2$, Inequality (\ref{eq:gammalpr}) follows from (\ref{eq:rho1}). 

By \cite{shannon}, for all rates
\[
\tilde{R} < \frac{1}{2}\log\frac{2+\sqrt{\tilde{P}^2/\sigma^4+4}}{4},
\]and for sufficiently large $n$
there exists a blocklength-$\tilde{n}$, rate-$\tilde R$  non-feedback coding scheme for the memoryless Gaussian point-to-point channel with noise variance $\sigma^2$, with expected average block-power no larger than $\tilde{P}$ and with probability of error  $P_e$ satisfying \begin{equation}\label{shan}
 P_e\leq e^{-\tilde n(E(\tilde{R},\tilde{P}/\sigma^2)-\epsilon')}
\end{equation}
for some fixed $\epsilon'>0$ and
\begin{equation}
 E(\tilde{R},\tilde{P})=\frac{\tilde{P}}{4\sigma^2}\left(1-\sqrt{1-e^{-2\tilde R}}\right).
\end{equation}
Since the probability of error of a non-feedback code over the Gaussian BC with common message is at most $K$ times the probability of error to the weakest receiver, we conclude that for all $\tilde P>0$ and
\begin{subequations}\label{eq:BCtildeR}
  \begin{IEEEeqnarray}{rCl}
  0<\tilde{R} < \frac{1}{2}\log\frac{2+\sqrt{\tilde{P}^2/\sigma_1^4+4}}{4},
 \end{IEEEeqnarray}
there exists a rate-$\tilde{R}$ code with power $\tilde{P}$ and blocklength $\tilde{n}$ that for the Gaussian BC with common message achieves probability of error
  \begin{IEEEeqnarray}{rCl}
  P^{\text{(BC)}}_{e}& \leq&       K e^{- \tilde{n}\left( \frac{P }{4  \sigma_1^2}\left(1-\sqrt{1-e^{-2 \tilde R}}\right)- \epsilon'\right)}.
 \end{IEEEeqnarray}
\end{subequations}
Now apply this statement to $\tilde{R}=R_{\text{phase},l}$, $\tilde{P}=P/\gamma_{l-1}$ and $\tilde{n}=\frac{\epsilon n}{L-1}-1$. Since for sufficiently large $n$, by (\ref{eq:gammalpr}),
\begin{IEEEeqnarray}{rCl}
 R_{\text{phase},l}< \frac{1}{2}\log\frac{2+\sqrt{\frac{P^2}{{\gamma^2_{l-1}}  \sigma_1^4}+4}}{4},
\end{IEEEeqnarray}
we conclude by (\ref{eq:BCtildeR}) that there exists a code $\set{C}_l$ of rate-$R_{\text{phase},l}$, block-power $P/\gamma_{l-1}$, blocklength $\frac{\epsilon n}{L-1}-1$ and probability of error $\rho_l$ satisfying
\begin{IEEEeqnarray}{rCl} \label{eq:aprhoell}
 \rho_l & \leq&       K e^{\!- \big(\!\frac{\epsilon n}{L\!-\!1}\!-\!1\!\big)\! \left(\! \frac{P }{4 {\gamma_{l\!-\!1}} \sigma_1^2}\Big(1\!-\!\sqrt{1\!-\!e^{\!-2 \!\frac{R(L-1)}{\epsilon\!-\!(L\!-\!1)/n}}}\!\Big)\!- \epsilon'\right)} \nonumber \\
 &\leq & \textnormal{exp}(-\underbrace{\textnormal{exp}\circ\ldots\circ\textnormal{exp}}_{l-1~\text{times}}(\Omega(n)))
\end{IEEEeqnarray}
where the inequality follows again by (\ref{eq:gammalpr}).

By the definition of $\gamma_l$ in (\ref{eq:gamrho}), Inequalities~(\ref{eq:aprhoell}) and \eqref{eq:gammalpr} also yield:
\begin{IEEEeqnarray}{rCl}\label{eq:gamrr}
 \gamma_{l}\leq  \textnormal{exp}(-\underbrace{\textnormal{exp}\circ\ldots\circ\textnormal{exp}}_{l-1~\text{times}}(\Omega(n))).
\end{IEEEeqnarray}

\end{document}